\documentclass[journal]{IEEEtran}   

\usepackage{amsmath,amssymb,amsthm}
\usepackage{graphicx}
\usepackage{cite}            
\usepackage{booktabs}        
\usepackage{caption,subcaption}
\usepackage{enumitem}        
\usepackage{url}             
\usepackage{algorithm,algpseudocode}

\usepackage{echodefs}
\newtheorem{theorem}{Theorem}
\newtheorem{lemma}{Lemma}
\newtheorem{proposition}{Proposition}

\newtheorem{corollary}{Corollary}
\newtheorem{assumption}{Assumption}
\usepackage{xcolor}
\usepackage{multirow}

\IEEEoverridecommandlockouts

\begin{document}


\title{
Chaotic Dependence and Degree-Structure Mechanisms in Optimal Pinning Control of Complex Networks
}

\author{Qingyang~Liu,~Tianlong~Fan\textsuperscript{*},~Liming~Pan\textsuperscript{*},~and~Linyuan~L\"{u}\textsuperscript{*}\IEEEmembership{, Senior Member, IEEE}
\thanks{The authors are with the School of Cyber Science and Technology, University of Science and Technology of China, 230026 Hefei, Anhui, China.}
\thanks{\textsuperscript{*}Corresponding authors: Tianlong Fan (tianlong.fan@ustc.edu.cn), Liming Pan (pan\_liming@ustc.edu.cn), Linyuan L\"{u} (linyuan.lv@ustc.edu.cn).}
\thanks{This work has been submitted to the IEEE for possible publication. 
Copyright may be transferred without notice, after which this version may no longer be accessible.}
}

\maketitle

\begin{abstract}

Selecting a minimal set of driver nodes to control large-scale complex networks is a central but computationally intractable problem, and the lack of general theory further hinders its resolution. Leveraging a degree-based mean-field (annealed) approximation from statistical physics, we analytically reveal how the structural degree distribution systematically governs synchronization performance and derive exact conditions for the globally optimal pinning set. Remarkably, the optimal configuration exhibits chaotic dependence: infinitesimal changes in the pinning budget can trigger abrupt and discontinuous shifts in both the composition of driver nodes and the resulting controllability. This phenomenon overturns the widely assumed monotonic improvement of greedy strategies and exposes an intrinsic structural instability in heterogeneous networks. Extensive simulations on synthetic and empirical systems confirm the universality of this effect and its governing role of degree-distribution features. These results establish a unified link between degree heterogeneity and network controllability, with implications for resilient infrastructures, neural dynamics, and social coordination. These insights further yield efficient low-complexity algorithms that approximate the global optimum, enabling practical applications in large-scale systems.

\end{abstract}

\begin{IEEEkeywords}
Network synchronization, pinning control, chaotic dependence, degree heterogeneity, annealed approximation
\end{IEEEkeywords}

\section{Introduction}

\IEEEPARstart{S}{ynchronization} of complex networks is a fundamental phenomenon underpinning diverse natural and engineered systems, from neuronal ensembles and power grids to social behaviors and biological rhythms~\cite{arenas2008synchronization}. However, spontaneous synchronization is rarely achievable in realistic network structures~\cite{pecora1998msf}, necessitating targeted external interventions, known as pinning control, wherein synchronization is driven by controlling only a strategically chosen subset of nodes—referred to as driver nodes.

Traditionally, node selection for pinning control relies predominantly on heuristic centralities, such as degree, betweenness, or closeness~\cite{grigoriev1997pinning,wang2002pinning,li2004pinning,liu2007pinning,zhou2008adaptive}. While intuitive, these heuristic strategies generally offer superficial node rankings without fully elucidating the underlying structural mechanisms driving synchronization. Advanced spectral approaches address this gap by analyzing eigenvalues of the Laplacian or grounded Laplacian matrices~\cite{pecora1998msf,moradiamani2017influential,yu2013optimalvertex,liu2024uniformity,liu2021spectralpinning,mao2025perturbation}, revealing deeper links between spectral properties and synchronization capability. Yet, despite their theoretical depth, these spectral methods typically remain computationally demanding and yield only locally optimal node sets, leaving the globally optimal selection unresolved.

Here, we propose a robust analytical framework to identify globally optimal pinning node sets by leveraging the annealed approximation from statistical physics~\cite{pastor2015epidemic,ferreira2012epidemic,pan2021optimal}. Under this approximation, we simplify the network structure by replacing the adjacency matrix entries with their expected connection probabilities derived from the Uniform Configuration Model (UCM)~\cite{pastor2015epidemic,ferreira2012epidemic,pan2021optimal}. This analytically tractable approach provides exact global solutions for the optimal driver-node selection, elucidating how synchronization performance is fundamentally governed by the structural degree distribution.

Remarkably, our analysis uncovers an unprecedented chaotic dependence phenomenon: incremental changes in the number of driver nodes can induce abrupt and dramatic shifts in the optimal node set composition and associated synchronization performance. This finding fundamentally challenges conventional top-$k$ greedy heuristics, which typically assume incremental improvement in node selection quality.

Through extensive numerical validations on both synthetic configuration models and real-world empirical networks, we systematically confirm that our analytically derived optimal solutions significantly outperform traditional heuristics. Furthermore, we systematically investigate how low-degree saturation, high-degree truncation, and the power-law exponent of network degree distributions affect synchronization effectiveness. In particular, our results identify low-degree saturation as an enhancement factor and high-degree truncation as a suppressive factor in synchronization efficiency, offering critical insights into the structural determinants of network controllability.

The remainder of this paper is structured as follows. In Section~\ref{sec:theory}, we rigorously develop the theoretical foundations for identifying the globally optimal pinning configuration under the annealed approximation. Section~\ref{sec:exp} presents extensive numerical validations, examining the effects of various degree distribution characteristics on synchronization and verifying theoretical predictions with empirical data. We conclude the paper in Section~\ref{sec:conclusion} by summarizing key contributions and outlining potential directions for future research.

\section{Theoretical Results}\label{sec:theory}

\subsection{Preliminaries and Problem Formulation}

\subsubsection{Network Dynamical Model With Pinning Control}

Consider a complex dynamical network composed of $N$ identical nodes with diffusive coupling, whose dynamics is governed by
\begin{equation}
    \dot{x}_i = \mathrm{f}(x_i) - \kappa \sum_{j=1}^N l_{ij} H x_j + u_i(x_1, \dots, x_N), 
\end{equation}
where $\quad i = 1, 2, \dots, N$ and $x_i \in \mathbb{R}^n$ is the state vector of node $i$, $\mathrm{f}(\cdot): \mathbb{R}^n \to \mathbb{R}^n$ represents the intrinsic dynamics of each node, $\kappa>0$ denotes the coupling strength, $u_i(\cdot)$ is the control input applied to node $i$, and $H \in \mathbb{R}^{n \times n}$ is a positive semi-definite inner coupling matrix. The matrix $\mathcal{L} = [l_{ij}]_{N \times N}$ is the Laplacian matrix of the network topology $\mathcal{G}=(\mathcal{V}, \mathcal{E})$, with $\mathcal{V}=\{1,2,\dots,N\}$ and $\mathcal{E}\subseteq\mathcal{V}\times\mathcal{V}$. For an undirected network, $\mathcal{L}$ is symmetric.

Let the target trajectory $s(t)$ satisfy
\begin{equation}
    \dot{s}(t) = \mathrm{f}(s(t)), \quad s(0) = s_0.
\end{equation}
The aim of pinning control is to select a subset of nodes to apply control such that all nodes synchronize to $s(t)$.

\begin{assumption}~\cite{belykh2004connection,chen2011bidirectionally,liu2013coupling,liu2015synchronization,wu2007synchronization} 
There exists a constant $\alpha>0$ such that
\begin{equation}
    (y - z)^{\top} \left[ \mathrm{f}(y) - \mathrm{f}(z) - \alpha H(y-z) \right] \leq -\mu \|y-z\|^2
\end{equation}
for some $\mu>0$ and all $y,z\in\mathbb{R}^n$. Here, $\alpha$ depends on the node dynamics $\mathrm{f}(\cdot)$ and the coupling matrix $H$.
\end{assumption}

Define the error $e_i(t) = x_i(t) - s(t)$. Let the set of pinned nodes be $P \subset \mathcal{V}$, with $|P| = c$; the remaining free (unpinned) nodes form the set $F = \mathcal{V} \setminus P$, with $|F| = N-c$. Without loss of generality, assume the first $c$ nodes are pinned. Then the error dynamics is
\begin{equation}
    \dot{e}_i = \mathrm{f}(x_i) - \mathrm{f}(s) - \kappa \sum_{j=1}^N l_{ij} H e_j + u_i, \quad i=1,\dots,N.
\end{equation}

A common linear feedback pinning controller is
\begin{equation}
u_i =
\begin{cases}
    -\kappa d H e_i, & i\in P,\\
    0, & i\in F,
\end{cases}
\end{equation}
where $d>0$ is the feedback gain.

Let $\mathcal{L}_F$ denote the $(N-c) \times (N-c)$ principal submatrix of $\mathcal{L}$ obtained by deleting the rows and columns corresponding to pinned nodes $P$. This $\mathcal{L}_F$ is the \emph{grounded Laplacian}. For undirected connected networks and nonempty $P$, $\mathcal{L}_F$ is symmetric positive definite.  

The pinning synchronization criterion~\cite{wang2008adaptive,song2009pinning,yu2009distributed} can be expressed as
\begin{equation}
    \lambda_1(\mathcal{L}_F) > \frac{\alpha}{\kappa},
\end{equation}
where $\lambda_1(\cdot)$ denotes the smallest eigenvalue. The quantity $\lambda_1(\mathcal{L}_F)$ thus serves as a \emph{pinning effectiveness metric} — the larger its value, the more effective the pinning scheme.  

Physically, $\lambda_1(\mathcal{L}_F)$ reflects the effective anchoring strength between the unpinned subnetwork $F$ and the pinned nodes $P$: removing $P$ and measuring the smallest eigenvalue of the remaining Laplacian quantifies how tightly the free nodes are connected to the pinned nodes.  

Therefore, for given $\kappa$ and $\alpha$, selecting a pinning set $P$ to maximize $\lambda_1(\mathcal{L}_F)$ leads to improved synchronization performance, and spectral optimization methods can be employed for node selection.

\subsubsection{Optimal Pinning under Annealed Approximation} 
In this section, we approach the problem analytically by leveraging the annealed approximation of network topology, which enables the derivation of a globally optimal solution.

Let $\calG=(V,E)$ be a graph, and $\calL$ be its Laplacian. An optimal pinning control strategy divide the nodes into two disjoint set as $V=P\cup F$, where $P$ corresponds to the set of pinning nodes and $F$ denotes the rest free nodes. For a set $P\subset V$ of pinning nodes, we denote the corresponding grounded Laplacian as $\calL_F$, which is defined as removing all rows and columns in $P$ from $\calL$ (or retaining only the rows the columns in $F$). 

The effectiveness of the pinning scheme can be measured by the smallest eigenvalue $\lambda_1(L_F)$ of the grounded Laplacian, where larger values indicate more effective control~\cite{liu2021spectralpinning,wang2008adaptive,song2009pinning,yu2009distributed}. The optimal pinning control problem can thus be formulated as
\begin{equation}
    \max_{P \subset V,\, |P|=c} \; \lambda_1(L_F).
    \label{eq:optimal_pinning}
\end{equation}

Direct optimization over $P$ to find a best grounded Laplacian can be difficult. Here we instead consider the optimization problem under the annealed approximation \cite{pastor2015epidemic}. Let the adjacency matrix of $\calG$ be $\mA$. Under the annealed approximation, we replace the adjacency matrix by the annealed one
\begin{equation}
    \mA=\frac{1}{K} \vd \vd^\intercal
\end{equation}
where the degree vector $\vd = (d_1,d_2,\dots,d_N)^\top$,
$d_i$ is the degree of node $i$ and $K = \sum_{n=1}^N d_n$ is the sum of all degrees. This approximation can be interpreted as assuming that the network is sampled from the uniform configuration model (UCM)~\cite{pastor2015epidemic}, and then replacing each edge with its corresponding probability under the UCM. The core idea is that node degrees capture the most essential structural properties of the network, while higher-order correlations between nodes are neglected. For this reason, the approach is also referred to as degree-based mean-field theory in statistical physics. Under the annealed approximation, the Laplacian matrix is defined accordingly based on these expected edge weights:
\begin{equation}
    \calL = \mD - \frac{1}{K} \vd \vd^\intercal,
\end{equation}
where $\mD \equiv \mathrm{diag}(\vd)$ is the diagonal matrix of all degrees.

\subsection{Optimal Pinning Strategy and Algorithmic Structure}

The main findings of the paper are summarized as follows. 

\begin{theorem}\label{theorem:algorithm}[Globally optimal pinning strategy]
Consider a connected network characterized by a strictly increasing node degree sequence $\hat{d}_1<\hat{d}_2<\cdots<\hat{d}_K$ with corresponding multiplicities $\gamma_1,\gamma_2,\dots,\gamma_K$. Define the cumulative node-count function as $\alpha(k)=\sum_{i=1}^{k}\gamma_i$. Then, given a prescribed pinning budget $c$, the globally optimal pinning set $P_c^*$ is uniquely determined by an integer threshold index $k_c^*\geq 0$ according to the following criteria:

\begin{enumerate}
    \item[(i)] \textbf{Optimality condition:} 
    The optimal pinning set $P_c^*$ always consists of two subsets: all $\alpha(k_c^*)$ nodes whose degrees are at most $\hat{d}_{k_c^*}$, together with the $c-\alpha(k_c^*)$ highest-degree nodes.

    \item[(ii)] \textbf{Special cases:} 
    If $c<\gamma_1$, the optimal threshold is always $k_c^*=0$, and the optimal set consists exclusively of the $c$ highest-degree nodes. In the particular scenario with $\gamma_1=1$ and $c=1$, the optimal choice is either the node with the highest degree or the single node with the lowest degree; explicitly comparing both outcomes determines the optimum.

    \item[(iii)] \textbf{Monotonicity and iterative construction:} 
    The threshold index $k_c^*$ is a monotonically non-decreasing function of $c$. Moreover, the optimal threshold at budget $c+1$ can be obtained iteratively from the threshold at budget $c$ as follows:

    \begin{itemize}
        \item[(a)] If $c+1<\alpha(k_c^*+1)$, the threshold remains unchanged: $k_{c+1}^*=k_c^*$.
        \item[(b)] If $c+1\geq\alpha(k_c^*+1)$, define $\Delta k$ as the minimal positive integer satisfying $c+1<\alpha(k_c^*+\Delta k)$. The optimal threshold $k_{c+1}^*$ is then determined by evaluating candidate pinning sets for all thresholds within the set $\{k_c^*,k_c^*+1,\dots,k_c^*+\Delta k-1\}$, and selecting the one that optimizes the smallest eigenvalue of the grounded Laplacian matrix.
    \end{itemize}
\end{enumerate}
\end{theorem}

Theorem~\ref{theorem:algorithm} characterizes the optimal pinning configuration as a degree-structured trade-off between low-degree and high-degree node selection. A key structural property is that low-degree nodes with identical degrees are grouped as indivisible units: for each degree level $\hat{d}_k$, the corresponding $\gamma_k$ nodes are either all included in the optimal solution or all omitted, at any given budget $c$. Moreover, once a particular degree level is included in the optimal set, it remains included for all larger values of $c$, due to the monotonicity of the threshold index $k_c^*$. This cumulative construction guarantees that the low-degree component of the optimal set has a nested, layered structure. As $c$ increases, $k_c^*$ may exhibit discontinuous jumps, leading to abrupt transitions in the composition of the optimal pinning set $P_c^*$. Such structural discontinuities give rise to a form of chaotic dependence on the budget $c$, whereby a marginal change in $c$ can cause a substantial shift in the selected nodes. These properties are empirically verified in Section~\ref{sec:exp}.

The remainder of this section is devoted to proving Theorem~\ref{theorem:algorithm}. Section~\ref{sec:properties} collects several useful properties of the grounded Laplacian and the optimal pinning strategy. In Section~\ref{sec:Boundedness and Monotonicity}, We first establish boundedness and monotonicity results, which provide the spectral foundations for the analysis. Section~\ref{sec:small}, address the case of small pinning sets, illustrating how the interplay between node degrees and spectral shifts leads to tractable optimality conditions. Finally, in Section~\ref{sec:general}, we extend these arguments to general pinning sets, thereby completing the theoretical framework. Section~\ref{sec:algorithm} then develops the algorithmic realization of these results and analyzes the computational complexity of the proposed strategy. For complete derivations and technical lemmas, see the Supplementary Material, Section ``Proofs of Theoretical Results.''

\subsection{Some Useful Properties}\label{sec:properties}
\subsubsection{Boundedness and Monotonicity}\label{sec:Boundedness and Monotonicity}
We assume the node degrees are arranged in non-decreasing order:
\[
d_1 \leq d_2 \leq \cdots \leq d_N.
\]
This degree sequence can be equivalently represented by the set of distinct degree values $\hat{d}_1 < \hat{d}_2 < \cdots < \hat{d}_Q$ and their corresponding multiplicities $\gamma_1, \gamma_2, \ldots, \gamma_Q$, where $\sum_{k=1}^Q \gamma_k = N$. Let $c = |P|$ denote the number of pinning nodes, and define $N' = |F| = N - c$ as the number of free nodes. After removing the pinning nodes, we relabel the nodes in $F$ via a mapping $v: [N'] \to [N]$ such that
\[
d_{v(1)} \leq d_{v(2)} \leq \cdots \leq d_{v(N')}.
\]
It is important to clarify that physically removing nodes from a network alters the degrees of the remaining nodes. However, in this context, node removal is interpreted as eliminating the corresponding rows and columns from the Laplacian matrix. Consequently, the degrees of the remaining nodes remain unchanged.

Let $\vd_F$ denote the degree vector of the free nodes, and define $\mD_F = \mathrm{diag}[\vd_F]$ as the diagonal matrix constructed from $\vd_F$. The grounded Laplacian is then given by:
\begin{equation}
    \mathcal{L}_F = \mathbf{D}_F - \frac{1}{K} \vd_F \vd_F^\intercal.
\end{equation}
By the Gershgorin circle theorem~\cite{gershgorin1931uber}, $\mathcal{L}_F$ is positive definite. Therefore, our primary interest lies in its smallest eigenvalue. In the following, we present several key properties of this eigenvalue.

\begin{lemma}\label{lemma:spectral}
Let $\lambda$ denote the smallest eigenvalue of the grounded Laplacian matrix associated with pinning set $P$ and free set $F$. Then $\lambda$ lies strictly within the interval $(0, \underline{d}_F)$, where $\underline{d}_F \equiv \min_{n \in F} d_n = d_{v(1)}$ denotes the minimum degree among the free nodes.

Moreover, $\lambda$ is the unique solution in the interval $(0, \underline{d}_F)$ to the following characteristic equation:
\begin{equation}\label{eqn:eig}
\sum_{n \in F} \frac{\lambda d_n}{d_n - \lambda} = \sum_{n \in P} d_n.
\end{equation}
\end{lemma}

Lemma~\ref{lemma:spectral} suggests that pinning low-degree nodes can be beneficial in certain scenarios, as the smallest nonzero eigenvalue of the grounded Laplacian is bounded above by $d_{v(1)}$. By pinning low-degree nodes, one may increase this upper bound, potentially improving the system performance.

However, this upper bound is not necessarily tight. In particular, when the pinning budget is insufficient to remove all nodes with the smallest degree—i.e., when $c < \gamma_1$—Lemma~\ref{lemma:spectral} does not offer guidance on which specific nodes to pin. 

We will show that, when the smallest-degree node in the free set $F$ is fixed, it is always more advantageous to pin high-degree nodes rather than additional low-degree ones. As a consequence of this result, when $c < \gamma_1$, the optimal pinning set consists of the $c$ highest-degree nodes.

\begin{theorem}\label{thm:swap}
Let $f \in F$ and $p \in P$ denote a free node and a pinned node, respectively, with corresponding degrees $d_f$ and $d_p$. Let $\underline{d}_F = \min_{n \in F} d_n$ denote the smallest degree among all free nodes. 

Suppose the following two conditions hold:
\begin{itemize}
    \item[(i)] $d_f > d_p \geq \underline{d}_F$;
    \item[(ii)] Removing node $f$ from $F$ does not alter the minimum degree $\underline{d}_F$ (i.e., $f$ is not the unique node attaining $\underline{d}_F$).
\end{itemize}

Then, exchanging the roles of $f$ and $p$—that is, pinning $f$ and freeing $p$—results in an increase in the smallest eigenvalue $\lambda$ of the grounded Laplacian.
\end{theorem}

An essential requirement in Theorem~\ref{thm:swap} is that the smallest degree in $F$ is fixed, which indicates that the upper bound from Lemma~\ref{lemma:spectral} do not changed by the swap. In other words, if the smallest degree in $F$ is fixed, pining the large degree nodes is always favorable. 

To concretize the implications of Theorem~\ref{theorem:algorithm}, we now present several corollaries characterizing the structure of optimal pinning sets under small control budgets.

\subsubsection{Small Pinning Set}\label{sec:small}
Theorem~\ref{thm:swap} has several immediate consequences. 

\begin{corollary}\label{coro:1}
Let $\gamma_1$ denote the multiplicity of the smallest degree $\hat{d}_1$ in the network. If the pinning budget satisfies $c < \gamma_1$, then the globally optimal strategy is to select the $c$ nodes with the highest degrees as the pinning set $P$.
\end{corollary}


\begin{corollary}\label{coro:2}
Let $\gamma_1$ and $\gamma_2$ denote the multiplicities of the first and second smallest degrees $\hat{d}_1$ and $\hat{d}_2$, respectively. Suppose the control budget satisfies $\gamma_1 \leq c < \gamma_1 + \gamma_2$. 

Then, any strategy that selects only $c_1 < \gamma_1$ nodes with degree $\hat{d}_1$ and assigns the remaining $c - c_1$ pins arbitrarily cannot be globally optimal.
\end{corollary}


From Corollary~\ref{coro:2}, when $\gamma_1\leq c < \gamma_1 + \gamma_2$, the only remaining choice is to either remove all nodes with degree $\hat{d}_1$ or none of them. In both cases, it is straightforward to show that choosing the rest undetermined nodes as large degree ones is optimal.

\begin{corollary}\label{coro:3}
Under the same setting where $\gamma_1 \leq c < \gamma_1 + \gamma_2$, the globally optimal pinning set must be one of the following two configurations:
\begin{itemize}
    \item[(a)] The $c$ highest-degree nodes in the network;
    \item[(b)] All $\gamma_1$ nodes with degree $\hat{d}_1$, together with the $c - \gamma_1$ highest-degree nodes.
\end{itemize}
\end{corollary}


For the two strategies presented in Corollary~\ref{coro:3}, we will demonstrate empirically in Section~\ref{sec:exp} that there appears to be no simple criterion to determine which strategy is superior. The optimal choice is highly sensitive to the specific numerical values of the degree sequence. Nevertheless, we can establish the following properties, which serve as a foundation for extending the analysis to general values of $c$. These properties are crucial for the proof of Theorem~\ref{theorem:algorithm}.

\begin{lemma}\label{lemma:switch}
Let $\gamma_1 \leq c < \gamma_1 + \gamma_2$, and consider the two candidate strategies described in Corollary~\ref{coro:3}. Strategy (b)—which includes all $\gamma_1$ nodes of degree $\hat{d}_1$ plus the $c - \gamma_1$ highest-degree nodes—yields a higher synchronizability (i.e., a larger grounded Laplacian eigenvalue $\lambda$) than strategy (a) if and only if the resulting $\lambda$ under strategy (b) satisfies
\[
\lambda \geq \hat{d}_1.
\]
\end{lemma}


Lemma~\ref{lemma:switch} reveals an important phenomenon: as the number of pinning nodes $c$ increases from $c < \gamma_1$ to $\gamma_1 \leq c \leq \gamma_1 + \gamma_2$, a transition may occur where the optimal strategy switches from (a) to (b). At this critical point, a ``chaotic'' dependence on $c$ may be observed, wherein a small change in $c$ can lead to a dramatic shift in the optimal pinning set.

Furthermore, Lemma~\ref{lemma:switch} implies that for any $\gamma_1 \leq c_1 < c_2 < \gamma_1 + \gamma_2$, if strategy (b) is optimal for $c_1$ pinning nodes, then it remains optimal for $c_2$. Indeed, using strategy (a) at $c_2$ can only yield a solution with eigenvalue $\lambda < \hat{d}_1$, while persisting with strategy (b) at $c_2$ results in an even larger eigenvalue $\lambda$ compared to the case with $c_1$ pins, as will be demonstrated through a simple observation.

\begin{proposition}\label{prop:1}
    When there are $c_1$ pinning nodes and suppose the pinning strategy is $(P_1,F_1)$, which results in an eigenvalue $\lambda_1$. Then for $c_2 =c_1+1$, construct a strategy by picking any nodes $f$ from $F_1$ let $P_2=P_1\cup \{f\}$ and $F_2 = F_1 \setminus \{f\}$. In addition, we assume the smallest degree in $F_1$ do not change by removing $f$. Let the resulting eigenvalue of strategy $(P_2,F_2)$ be $\lambda_2$, then $\lambda_2 > \lambda_1$.
\end{proposition}


\subsubsection{General Pinning Set}\label{sec:general}

Corollary~\ref{coro:2} suggests if we decide to pin a node with a small degree $\hat{d}_1$, the optimal strategy can only be pin all nodes with the same degree collectively.  A similar conclusion can be drawn for general $c$.

\begin{corollary} \label{coro:2_general}
Let $c < N$ be the number of pinned nodes, and consider a pinning strategy $(P, F)$. Suppose there exist two nodes $p \in P$ and $f \in F$ such that $d_p < d_f$. Let $k$ be the index such that $\hat{d}_k = d_p$, and define $c_k = \vert \{ n \in P : d_n = \hat{d}_k \} \vert$, i.e., the number of pinned nodes with degree $\hat{d}_k$. If $c_k < \gamma_k$ (i.e., not all nodes with degree $\hat{d}_k$ are included in $P$), then the strategy cannot be globally optimal.
\end{corollary}


Corollary~\ref{coro:2_general} states that if a low-degree node (in the sense there exists $d_f>d_p$ for $p\in P$ and $f\in F$), then all nodes with the same degree as $p$ must also be pinned. Moreover, we can show that if all nodes with degree $\hat{d}_k$ are pinned, this can only be optimal if all nodes with degrees strictly less than $\hat{d}_k$ are also included in the pinning set.

\begin{corollary} \label{coro:2_general2}
Let $c < N$ and consider a pinning strategy $(P, F)$. Suppose that all nodes with degree $\hat{d}_k$ are included in $P$ for some $k$. If there exist two nodes $f, f' \in F$ such that $d_{f'} < \hat{d}_k < d_f$, then the strategy cannot be globally optimal.
\end{corollary}


With Corollary~\ref{coro:2_general2}, we can conclude that the optimal pinning strategy must take the following form: the optimal pinning set $P$ contains all low-degree nodes up to index $k_c$; then, by Theorem~\ref{thm:swap}, the remaining nodes must be the highest-degree ones. However, it remains unclear how $k_c$ varies with $n$. To address this, we extend Lemma~\ref{lemma:switch} to the general case of $c$ as follows.

\begin{lemma}\label{lemma:dominance}
Let $c = \vert P \vert$ be the size of the pinning set, and suppose that for some integer $k$, the cumulative multiplicity satisfies $\sum_{i=1}^{k} \gamma_i < c$. Consider the following two pinning strategies:
\begin{itemize}
    \item[(a)] Select all nodes with degrees $\hat{d}_1, \ldots, \hat{d}_{k-1}$ (i.e., $\sum_{i=1}^{k-1} \gamma_i$ nodes), and choose the remaining $c - \sum_{i=1}^{k-1} \gamma_i$ nodes from those with the highest degrees.
    \item[(b)] Select all nodes with degrees $\hat{d}_1, \ldots, \hat{d}_{k}$ (i.e., $\sum_{i=1}^{k} \gamma_i$ nodes), and choose the remaining $c - \sum_{i=1}^{k} \gamma_i$ nodes from those with the highest degrees.
\end{itemize}
Let $\lambda_a$ and $\lambda_b$ denote the smallest eigenvalues of the grounded Laplacian resulting from strategies (a) and (b), respectively. Then strategy (b) dominates strategy (a)—i.e., yields a larger $\lambda$—if and only if
\[
\lambda_b \geq \hat{d}_k > \lambda_a.
\]
\end{lemma}

With all the preparations, we are ready to complete the proof of Theorem~\ref{theorem:algorithm}.
\begin{proof}[Proof of Theorem~\ref{theorem:algorithm}]
When $c < \gamma_1$, the optimality of procedure $\mathcal{P}$ in Theorem~\ref{theorem:algorithm} is established in Corollary~\ref{coro:1}. The special case where $\gamma_1 = 1$ and $c = 1$ is included in Corollary~\ref{coro:3}, completing the proof of (\romannumeral2) in Theorem~\ref{theorem:algorithm}.
From Corollary~\ref{coro:2_general2} and the subsequent discussions, we can infer that the optimal strategy for a fixed $c$ must follow the procedure $\mathcal{P}$ in Theorem~\ref{theorem:algorithm}, thereby confirming (\romannumeral1).

Next, we prove that $k^*_c$ is non-decreasing with respect to $c$. Consider the case with $c$ pinning nodes. Let $k$ in Lemma~\ref{lemma:dominance} be set to $k_c^*$. Since strategy (b) is globally optimal, it dominates strategy (a), and thus $\lambda_c \geq \hat{d}_{k_c^*}$. Now, consider the case with $c+1$ pinning nodes. Suppose $P_{c+1}$ excludes any node with degree equal to or below $\hat{d}_{k_c^*}$. By Lemma~\ref{lemma:spectral}, the resulting eigenvalue satisfies $\lambda_{c+1} < \hat{d}_{k_c^*} \leq \lambda_c$. However, by Proposition~\ref{prop:1}, we can always find a solution where $\lambda_{c+1} > \lambda_c$ by insisting on pinning all nodes with degree equal to or below $\hat{d}_{k_c^*}$, as in the case for $c$ pinning nodes. Therefore, the optimal strategy for $c+1$ pinning nodes must include all such nodes.

Since $k^*_c$ is non-decreasing in $c$, we can construct the solution $P_{c+1}^*$ from $P_c^*$. If $c+1 < \alpha(k^*_c+1)$, then the smallest degree in the free set $F$ cannot be increased, so $k^*_{c+1} = k^*_c$, confirming part (a) in (\romannumeral3). Otherwise, we need to examine all possible values of $k_{c+1}$ and select the optimal one, which corresponds to the procedure described in part (b) of (\romannumeral3).
\end{proof}

\subsection{Algorithm and Complexity Analysis}\label{sec:algorithm}

We first introduce the baseline method, Algorithm 1(denoted as $\mathcal{A}_1$), which implements a \textbf{low-degree-prioritized} selection rule. 
To construct a pinning set of size $c$, $\mathcal{A}_1$ identifies the largest index $k$ such that 
\[
\sum_{i=1}^{k} \gamma_i \leq c \quad \text{and} \quad \sum_{i=1}^{k+1} \gamma_i > c.
\] 
It then selects all nodes with degrees less than or equal to $\hat{d}_k$ and fills the remaining budget by choosing the highest-degree nodes. 
This procedure requires only two linear scans through the sorted degree sequence, resulting in an overall near-linear complexity of $\mathcal{O}(N+M)$.

While $\mathcal{A}_1$ is computationally efficient, it follows a fixed degree-prioritized rule that may not fully capture the trade-off between low- and high-degree node selection identified in Section~\ref{sec:theory}. 
As discussed there, the optimal pinning set under the annealed approximation arises from a balance: selecting low-degree nodes raises the upper bound of the smallest grounded Laplacian eigenvalue $\lambda$, whereas selecting high-degree nodes maximizes $\lambda$ once the bound is set.

To overcome this limitation, we propose Algorithm 2 (denoted as $\mathcal{A}_2$), whose pseudocode is summarized in the same section. 
$\mathcal{A}_2$ employs an iterative update scheme to determine the optimal threshold index $k^*_c$ for each budget level $c \in [1, c_{\max}]$, as prescribed in Theorem~\ref{theorem:algorithm}, and evaluates candidate solutions by computing the smallest nonzero eigenvalue $\lambda_1$ of the grounded Laplacian. 
Considering both preprocessing and eigenvalue evaluation, the overall runtime of $\mathcal{A}_2$ is approximately $\tilde{\mathcal{O}}(c_{\max} N)$, maintaining near-linear scalability for sparse networks. For complete pseudocode listings and the detailed time-complexity analysis of both $\mathcal{A}_1$ and $\mathcal{A}_2$, see the Supplementary Material, Section ``Pseudocode and Complexity Analysis of Algorithms''

Comparing $\mathcal{A}_1$ and $\mathcal{A}_2$ enables us to examine how the transition between low-degree and high-degree regimes influences the formation of the optimal pinning configuration.

\section{Numerical Experiments}\label{sec:exp}
\subsection{Validation on Small Networks via Exhaustive Enumeration}\label{sec:validation}

\begin{figure*}[htbp]
    \centering
    \begin{subfigure}{0.4\textwidth}
        \centering
        \includegraphics[width=\textwidth, height=0.73\textwidth]{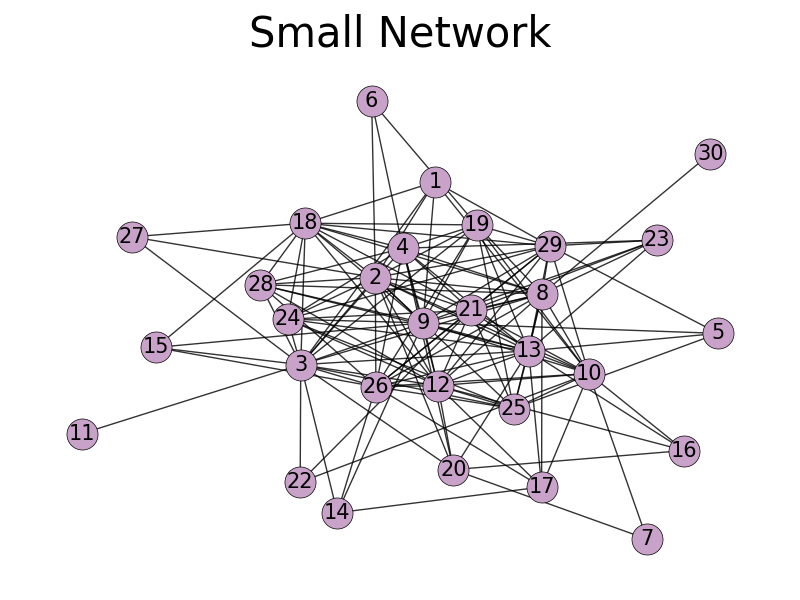}
        \caption{}
    \end{subfigure}
    \begin{subfigure}{0.4\textwidth}
        \centering
        \includegraphics[width=\textwidth, height=0.73\textwidth]{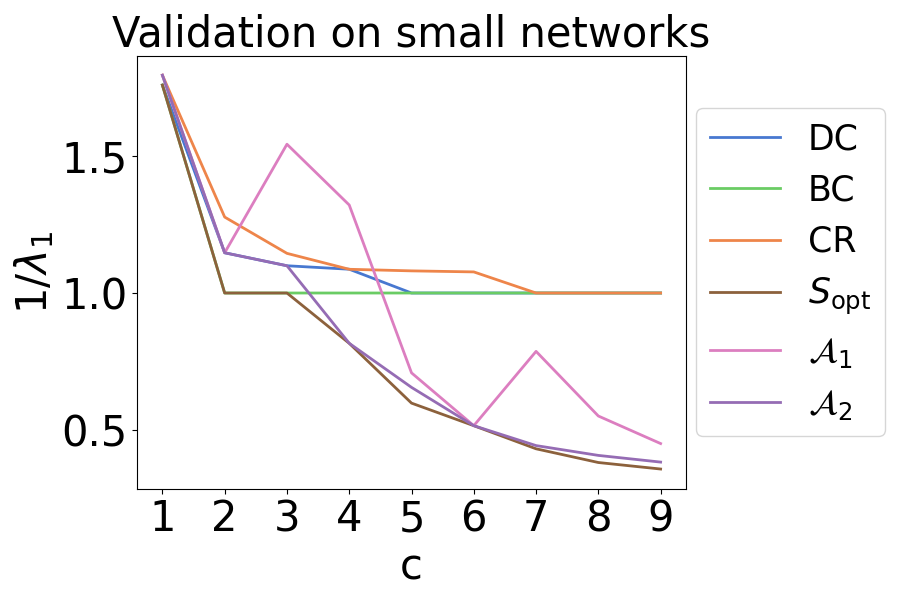}
        \caption{}
    \end{subfigure}
    \caption{Validation of analytical algorithms via exhaustive enumeration on a small configuration model. The global optimum ($S_{\mathrm{opt}}$) is obtained exactly for each cardinality $c \in \{1, \dots, 9\}$. Both $\mathcal{A}_1$ and $\mathcal{A}_2$ accurately recover the optimum across nearly all values of $c$, whereas all three centrality-based heuristics (DC/BC/CR) fail to consistently match optimal performance.}    
    \label{fig:toy_model}
\end{figure*}

\begin{table*}[htbp]
    \caption{Chaotic dependence of the optimal pinning sets for $S_{\mathrm{opt}}$ and $\mathcal{A}_2$ on the budget $c$. Bold entries denote nodes appearing at size $c$ but not in the corresponding set of size $c-1$. The Hamming distance $d_{\mathrm{hm}}$ quantifies the change in composition, and \(1/\lambda_1\) is the corresponding inverse synchronizability index.}
    \centering
    \resizebox{0.9\textwidth}{!}{
    \begin{tabular}{c|ccc|ccc}
        \toprule
        \multirow{2}{*}{\(\mathbf{c}\)} 
            & \multicolumn{3}{c|}{$S_{\mathrm{opt}}$} 
            & \multicolumn{3}{c}{$\mathcal{A}_2$} \\
        \cmidrule(lr){2-4} \cmidrule(lr){5-7}
            & \(\mathbf{S}\) & \(\mathbf{d_{hm}}\) & \(\mathbf{1/\lambda_1}\) 
            & \(\mathbf{S}\) & \(\mathbf{d_{hm}}\) & \(\mathbf{1/\lambda_1}\) \\
        \midrule
        1 & ($||$ 3) & $-$ & 1.7592 & ($|$ 13)& $-$ & 1.7953 \\
        2 & ($||$ \textbf{8}, 3) & 1 & 1.0000 & ($|$ \textbf{3}, 13) & 1 & 1.1469 \\
        3 & (\textbf{30} $||$ \textbf{12}, 3) & 3 & 1.0000 & ($|$ \textbf{2}, 3, 13) & 1 & 1.0995 \\
        4 & (30, \textbf{11} $||$ \textbf{2}, \textbf{13}) & 5 & 0.8163 & (\textbf{11}, \textbf{30} $|$ 3, 13) & 3 & 0.8164 \\
        5 & (11, 30 $||$ \textbf{10}, \textbf{3}, 13) & 3 & 0.5980 & (11, 30 $|$ \textbf{2}, 3, 13) & 1 & 0.6555 \\
        6 & (11, 30, \textbf{7} $||$ \textbf{2}, 3, 13) & 3 & 0.5156 & (11, 30, \textbf{7} $|$ 2, 3, 13) & 1 & 0.5156 \\
        7 & (11, 30, 7 $||$ \textbf{9}, 2, 3, 13) & 1 & 0.4312 & (11, 30, 7 $|$ \textbf{12}, 2, 3, 13) & 1 & 0.4429 \\
        8 & (11, 30, 7 $||$ \textbf{4}, \textbf{10}, 9, 2, 3) & 3 & 0.3811 & (11, 30, 7 $|$ \textbf{8}, 12, 2, 3, 13) & 1 & 0.4071 \\
        9 & (11, 30, 7 $||$ \textbf{18}, 4,  10, 2, 3, \textbf{13}) & 3 & 0.3573 & (11, 30, 7 $|$ \textbf{9}, 8, 12, 2, 3, 13) & 1 & 0.3826 \\
        \bottomrule
    \end{tabular}}
    \label{tab:optimal_set_of_toy_model}
\end{table*}

Since Theorem~\ref{theorem:algorithm} is derived under the annealed approximation, it is crucial to assess its validity on small networks where the globally optimal pinning set can be computed exactly. To this end, we perform exhaustive enumeration on configuration model networks with $N = 30$ nodes and a power-law degree distribution exhibiting low-degree saturation (see Section~\ref{section:power_law_features} for details). For this toy network, we determine the exact optimal pinning set for each target cardinality $c \in \{1,2,\dots,0.3N\}$ by evaluating all possible node combinations.  

We compare the performance of our proposed algorithms $\mathcal{A}_1$ and $\mathcal{A}_2$ against three representative centrality baselines: Degree Centrality (DC), Betweenness Centrality (BC), and Cycle Ratio (CR)~\cite{fan2021cycles}. Each baseline forms the pinning set by selecting the top-$c$ nodes ranked by its respective centrality measure. The synchronizability is quantified using $1/\lambda_1$~\cite{fan2021cycles}, which serves as an inverse synchronizability index, with smaller values indicating stronger post-pinning synchrony.  

Figure~\ref{fig:toy_model}(a) illustrates the toy network with $N=30$ nodes, while Fig.~\ref{fig:toy_model}(b) reports the validation results. We observe that $\mathcal{A}_1$ and $\mathcal{A}_2$ consistently recover the exact global optimum across most values of $c$, while substantially outperforming all three centrality-based heuristics. Notably, in this small-scale setting, $\mathcal{A}_1$ and $\mathcal{A}_2$ yield nearly identical performance to the global optimum, whereas DC, BC, and CR fail to consistently identify the true optimal sets. These results demonstrate the accuracy and robustness of our theoretical framework and provide a strong empirical basis for applying it to larger networks under the annealed approximation.  

To further examine the sensitivity of pinning sets to budget changes, Table~\ref{tab:optimal_set_of_toy_model} reports both the exact global optimum $S_{\mathrm{opt}}$ and the algorithmic solution $\mathcal{A}_2$. For each $c$, we list the selected set, the Hamming distance $d_{\mathrm{hm}}$ from the previous set, and the corresponding \(1/\lambda_1\).  The Hamming distance is defined as
\[
\begin{aligned}
d_{\mathrm{hm}}(c) = | \big(S_{\mathrm{opt}}(c) \setminus S_{\mathrm{opt}}(c-1)\big) \\
                     \cup \big(S_{\mathrm{opt}}(c-1) \setminus S_{\mathrm{opt}}(c)\big) |,
\end{aligned}
\]
which counts the number of nodes replaced or added when $c$ increases by one. To highlight structural differences, we use the symbols $||$ and $|$ to partition nodes into low- and high-degree groups: in $S_{\mathrm{opt}}$, the split $||$ is applied manually, while in $\mathcal{A}_2$, the split $|$ emerges naturally from the algorithm. This representation reveals that both methods capture the same essential principle in Lemma~\ref{lemma:switch}: low-degree nodes determine the synchronizability boundary, while high-degree nodes complement the set.  

From the table we find that both $S_{\mathrm{opt}}$ and $\mathcal{A}_2$ exhibit frequent cases of $d_{\mathrm{hm}}>1$, confirming the ``chaotic dependence'' of pinning sets on budget $c$. At the same time, the selection of low-degree nodes is highly consistent between the two methods, explaining why their synchronizability curves in Fig.~\ref{fig:toy_model}(b) are nearly identical. In contrast, the heuristic baselines overlook this low-degree constraint and therefore achieve inferior performance. These findings not only validate our theory but also clarify the mechanism by which the proposed algorithm aligns with the global optimum.

\subsection{Evaluating Optimal Pinning Strategies across Networks with Varying Degree-Distribution Features}\label{section:power_law_features}

\begin{figure}[htbp]
  \centering
  \includegraphics[width=0.8\columnwidth]{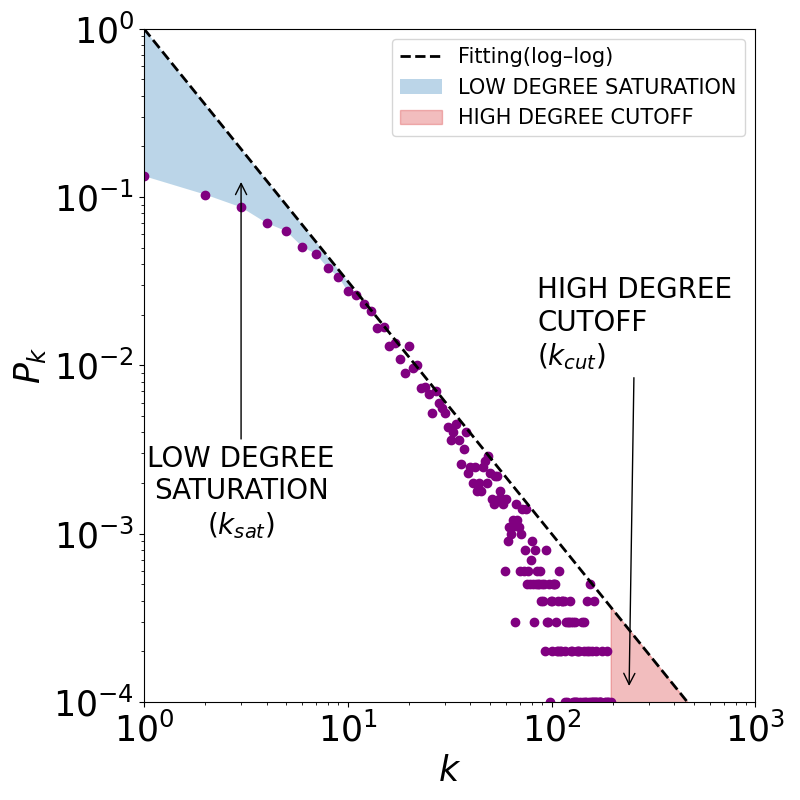}
  \caption{Representative degree distributions illustrating two common deviations from an ideal power-law: low-degree saturation and high-degree cutoff.}
  \label{fig:1}
\end{figure}

\begin{figure*}[htbp]
    \centering
    \begin{subfigure}{0.3\textwidth}
        \centering
        \includegraphics[width=\textwidth, height=0.73\textwidth]{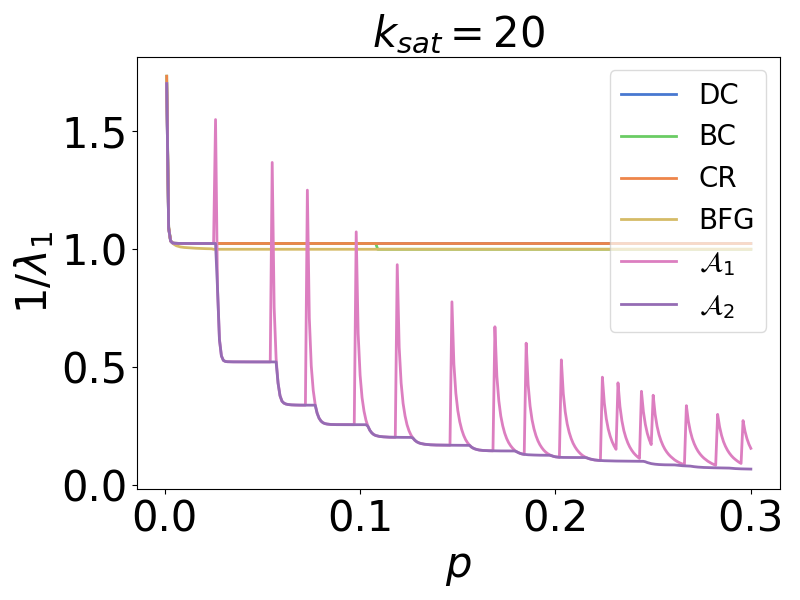}
        \caption{}
    \end{subfigure}
    \begin{subfigure}{0.3\textwidth}
        \centering
        \includegraphics[width=\textwidth, height=0.73\textwidth]{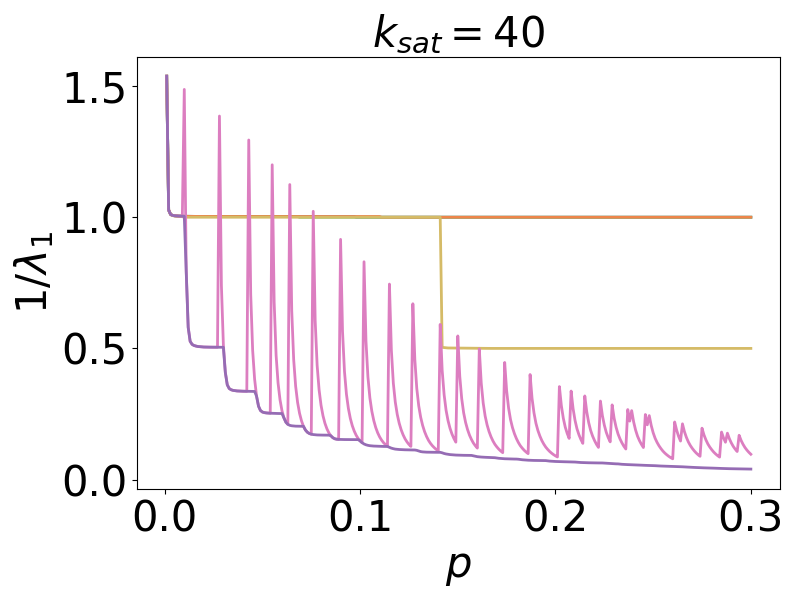}
        \caption{}
    \end{subfigure}
    \begin{subfigure}{0.3\textwidth}
        \centering
        \includegraphics[width=\textwidth, height=0.73\textwidth]{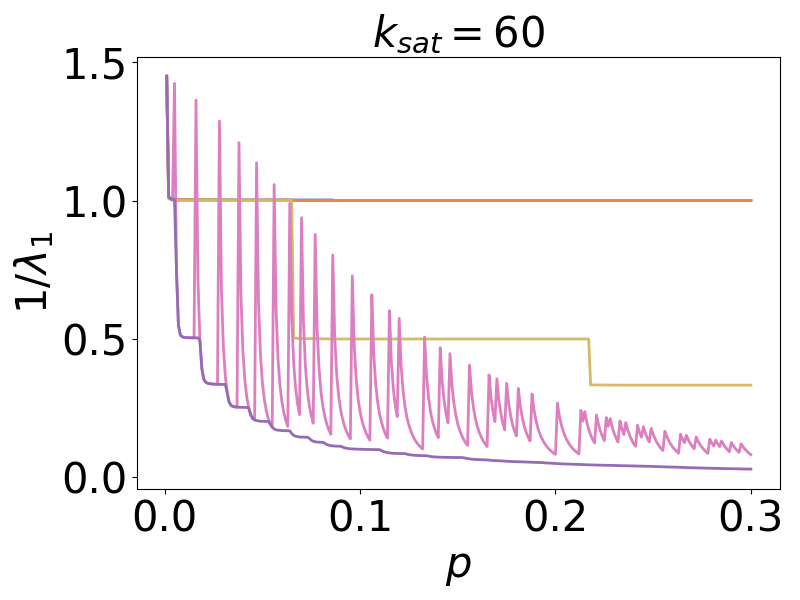}
        \caption{}
    \end{subfigure}
    \begin{subfigure}{0.3\textwidth}
        \centering
        \includegraphics[width=\textwidth, height=0.73\textwidth]{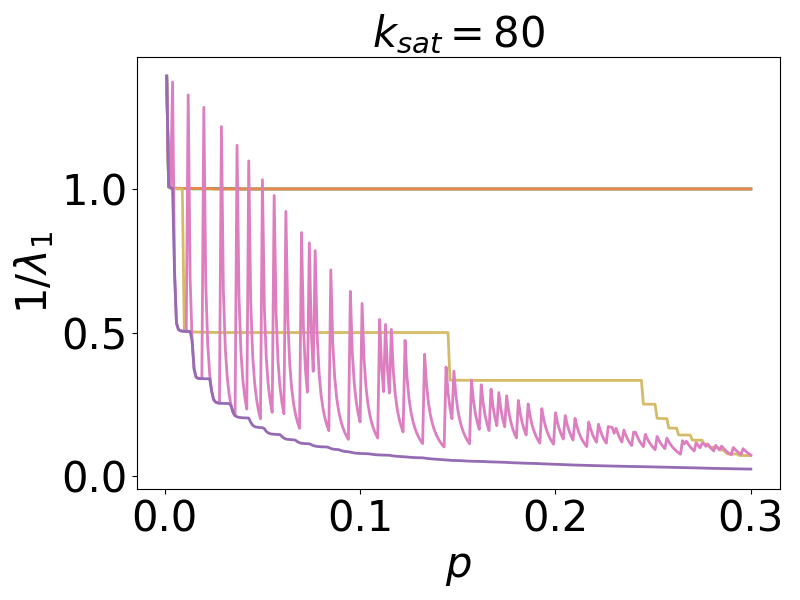}
        \caption{}
    \end{subfigure}
    \begin{subfigure}{0.3\textwidth}
        \centering
        \includegraphics[width=\textwidth, height=0.73\textwidth]{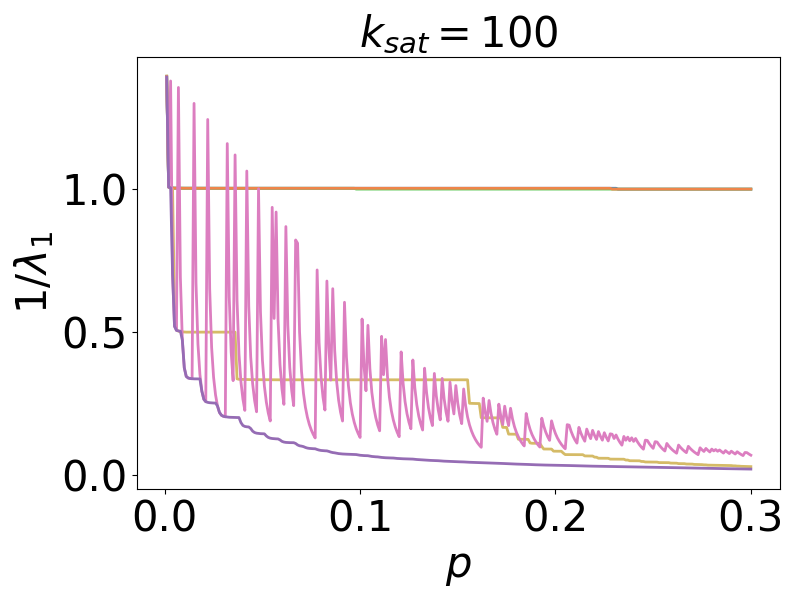}
        \caption{}
    \end{subfigure}
    \begin{subfigure}{0.3\textwidth}
        \centering
        \includegraphics[width=\textwidth, height=0.73\textwidth]{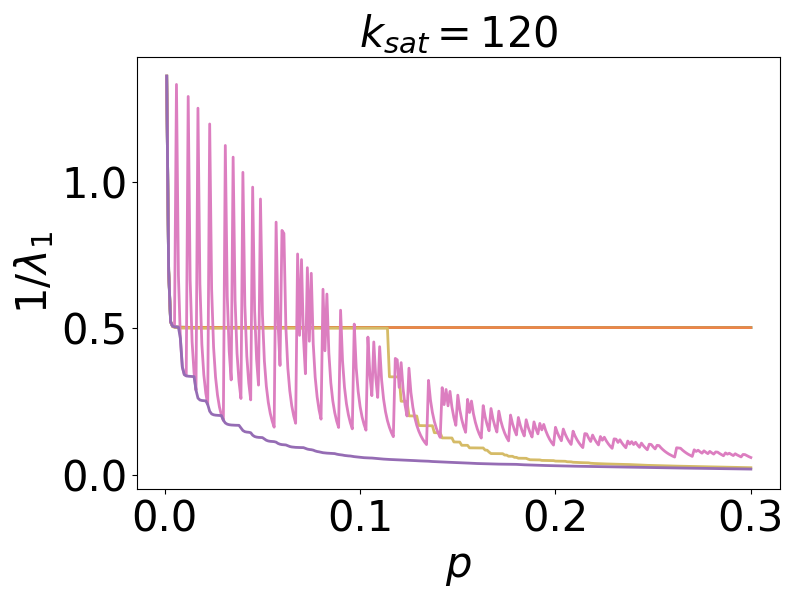}
        \caption{}
    \end{subfigure}
    \caption{Effect of low-degree saturation on pinning performance for configuration networks with \(k_{\text{sat}}=20,40,60,80,100,120\). The horizontal axis shows the pinned-node fraction \(p\) (0–0.3). Curves show \(1/\lambda_1(L_{N(1-p)})\) for various strategies (BC: Betweenness Centrality, CC: Coreness Centrality, DC: Degree Centrality, CR: Cycle Ratio, BFG: brute-force greedy strategy, \(\mathcal{A}_1\), \(\mathcal{A}_2\)); smaller values indicate better synchronizability.}    
    \label{fig:results1}
\end{figure*}

\begin{table*}[htbp]
    \caption{Pinning Efficiency \( \omega \) and End-Point Pinning Effectiveness \( \delta \) under varying low-degree saturation \( k_{\text{sat}} \). Best results are in bold.}
    \centering
    \resizebox{\textwidth}{!}{
    \large 
    \renewcommand{\arraystretch}{1.2}
    \begin{tabular}{c|cc|cc|cc|cc|cc|cc|c|c}
        \toprule
        \multirow{2}{*}{$k_{sat}$} & \multicolumn{2}{c|}{DC} & \multicolumn{2}{c|}{BC} & \multicolumn{2}{c|}{CR} & \multicolumn{2}{c|}{BFG} & \multicolumn{2}{c|}{$\mathcal{A}_1$} & \multicolumn{2}{c|}{$\mathcal{A}_2$} & \multicolumn{1}{c|}{$\Delta_{\omega}$} & \multicolumn{1}{c}{$\Delta_{\delta}$} \\
        & $\omega$ & $\delta$ & $\omega$ & $\delta$ & $\omega$ & $\delta$ & $\omega$ & $\delta$ & $\omega$ & $\delta$ & $\omega$ & $\delta$ & & \\
        \midrule
        20  & 1.0269 & 1.0244 & 1.0114 & 1.0000 & 1.0271 & 1.0244 & 1.0033 & 1.0000 & 0.3437 & 0.1560 & \textbf{0.2768} & \textbf{0.0680} & \textbf{72.41\%} & \textbf{93.20\%} \\
        40  & 1.0028 & 1.0000 & 1.0025 & 1.0000 & 1.0029 & 1.0000 & 0.7371 & 0.5000 & 0.2927 & 0.0966 & \textbf{0.1743} & \textbf{0.0401} & \textbf{76.35\%} & \textbf{91.98\%} \\
        60  & 1.0022 & 1.0000 & 1.0018 & 1.0000 & 1.0020 & 1.0000 & 0.5638 & 0.3333 & 0.2767 & 0.0822 & \textbf{0.1313} & \textbf{0.0300} & \textbf{76.71\%} & \textbf{91.00\%} \\
        80  & 1.0015 & 1.0000 & 1.0014 & 1.0000 & 1.0014 & 1.0000 & 0.3938 & 0.0716 & 0.2709 & 0.0733 & \textbf{0.1104} & \textbf{0.0244} & \textbf{71.96\%} & \textbf{65.92\%} \\
        100 & 1.0035 & 1.0000 & 1.0022 & 1.0000 & 1.0035 & 1.0000 & 0.2405 & 0.0298 & 0.2659 & 0.0698 & \textbf{0.0953} & \textbf{0.0212} & \textbf{60.37\%} & \textbf{28.86\%} \\
        120 & 0.5058 & 0.5021 & 0.5058 & 0.5021 & 0.5058 & 0.5021 & 0.2393 & 0.0237 & 0.2588 & 0.0585 & \textbf{0.0826} & \textbf{0.0190}& \textbf{65.48\%} & \textbf{19.83\%} \\
        \bottomrule
    \end{tabular}}
    \label{tab:results_ksat}
\end{table*}

Theorem~\ref{theorem:algorithm} establishes that, under the annealed approximation, the optimal pinning set is fully determined by the degree sequence, in particular by the abundance of low- and high-degree nodes. Motivated by this result, we design systematic experiments to verify that the proposed algorithm can accurately construct optimal pinning sets for networks with diverse degree-distribution characteristics. Specifically, we consider three representative features commonly observed in real-world networks—low-degree saturation, high-degree cutoff~\cite{barabasi2013networkscience, albert1999diameter} (see Figure~\ref{fig:1}), and variation in the power-law exponent:  
\begin{itemize}
    \item \textbf{Low-degree saturation} — A flattening of the degree distribution $p_k$ for $k < k_{\text{sat}}$, resulting in fewer very low-degree nodes than predicted by an ideal power law. Such saturation commonly arises from structural constraints, such as design requirements in engineered systems or minimum connectivity thresholds in infrastructure and communication networks.

    \item \textbf{High-degree cutoff} — An accelerated decay of $p_k$ for $k > k_{\text{cut}}$, which limits the size of the largest hubs. This phenomenon is often driven by intrinsic restrictions on node connectivity; for example, in social systems, individuals face cognitive and cost limits that cap the number of stable relationships, whereas in technological systems, bandwidth or hardware capacity constraints can produce similar effects.

    \item \textbf{Power-law exponent $\boldsymbol{\gamma}$} — The slope of the tail of the degree distribution, which controls the degree of heterogeneity. Smaller $\gamma$ yields a flatter tail with a heavier proportion of hubs, whereas larger $\gamma$ produces a steeper tail and more homogeneous connectivity. The value of $\gamma$ is shaped by network growth dynamics and attachment mechanisms, such as preferential attachment or resource competition.
\end{itemize}

These three features jointly determine the degree distribution, which can be modeled as~\cite{barabasi2013networkscience, albert1999diameter}
\begin{equation}\label{eq:degree_distribution}
p_k = a(k + k_{\text{sat}})^{-\gamma} \exp\left(-\frac{k}{k_{\text{cut}}}\right),
\end{equation}
where $k_{\text{sat}}$ captures low-degree saturation, the exponential term models high-degree cutoff, and $\gamma$ controls degree heterogeneity.

In the following experiments, we generate synthetic networks using the uniform configuration model with degree distributions specified by Equation~\eqref{eq:degree_distribution}. In each experiment, we vary a single parameter while keeping the other two fixed, thus enabling a controlled assessment of how each structural characteristic influences the optimal pinning performance. Here we present and discuss the results for the first two parameters; the experiments concerning the power-law exponent $\gamma$ are deferred to the Supplementary Material, Section ``Experimental Results of the Power-law Exponent.''

\subsubsection{Low-degree saturation}

We first isolate the effect of low-degree saturation by setting the high-degree cutoff to infinity. Specifically, we modify the degree distribution in Equation~\eqref{eq:degree_distribution} to
\[
p_k = a(k + k_{\text{sat}})^{-\gamma},
\]
with parameters \( a = 1 \) and \( \gamma = 1.5 \). We compare four baseline node selection strategies—degree centrality(DC), betweenness centrality(BC), and cycle ratio(CR)~\cite{fan2021cycles} and brute-force greedy strategy(BFG, a $1/\lambda_1(\mathcal{L}_F)$-based brute greedy selector)—against our two proposed methods: Algorithm \( \mathcal{A}_1 \), which exhibits oscillatory behavior in the effectiveness curve, and Algorithm \( \mathcal{A}_2 \), which eliminates this oscillation and yields globally optimal solutions. Notably, BFG achieves state-of-the-art performance among heuristic baselines, albeit at substantially higher computational cost.

We use the same performance metric, \(1/\lambda_1(\mathcal{L}_F)\), as defined in Section~\ref{section:power_law_features}, where \(p\) denotes the fraction of pinned nodes; smaller values of this metric indicate stronger control effectiveness~\cite{liu2021spectralpinning,fan2021cycles}. Experiments are conducted on networks with six low-degree saturation levels: \(k_{\text{sat}} = 20, 40, 60, 80, 100, 120\).

As shown in Figure~\ref{fig:results1}, both proposed algorithms consistently outperform all baselines. Algorithm \( \mathcal{A}_1 \) exhibits pronounced oscillations: when the number of selected nodes matches the size of a specific low-degree layer, the strategy switches abruptly from targeting locally beneficial nodes to selecting the entire layer. This discrete transition causes sharp spikes in \( 1/\lambda_1 \) (indicating a temporary drop in control effectiveness), followed by gradual decreases as more beneficial nodes are pinned, producing a characteristic sawtooth pattern. Algorithm \( \mathcal{A}_2 \) eliminates these oscillations by refining the switching rule: it only transitions to selecting an entire low-degree layer when doing so yields a strictly smaller \( 1/\lambda_1(\mathcal{L}_F) \) than selecting the same number of high-degree nodes, thereby consistently producing globally optimal pinning configurations with smooth performance curves.

Through our earlier descriptions of Algorithm \( \mathcal{A}_1 \) and Algorithm \( \mathcal{A}_2 \), as well as small-network validations in Section~\ref{sec:validation}, we observe that the solutions produced by both algorithms exhibit a “chaotic dependence” on the number of pinned nodes $c$. However, since Algorithm \( \mathcal{A}_2 \) enjoys optimality guarantees while Algorithm \( \mathcal{A}_1 \) does not, the curve corresponding to Algorithm \( \mathcal{A}_1 \) in pinning experiments tends to oscillate at those values of $c$ where “chaotic dependence” occurs. As the number of pinned nodes increases, such oscillations typically appear multiple times in the experiments, whereas the solutions of Algorithm \( \mathcal{A}_2 \) do not.

Furthermore, in Figure~\ref{fig:results1}, increasing \( k_{\text{sat}} \) raises the oscillation frequency of Algorithm \( \mathcal{A}_1 \) and reduces the horizontal width of each oscillation cycle. Within each curve, the vertical amplitude of oscillations tends to diminish as \( p \) increases, reflecting smoother performance at larger pinning fractions. Both algorithms show improved overall performance with higher \( k_{\text{sat}} \), consistent with our theoretical prediction that stronger low-degree saturation facilitates more effective pinning by enabling efficient selection of low-degree layers and increasing the minimum degree of remaining nodes.

We also report two complementary metrics in Table~\ref{tab:results_ksat}: the \emph{Pinning Efficiency}~\cite{fan2021cycles}, denoted by \( \omega \), and the \emph{End-point Pinning Effectiveness}, denoted by \( \delta \). The metric \( \omega \) is defined as
\begin{equation}
    \omega = \frac{1}{c_{\max}} \sum_{n=1}^{c_{\max}} \frac{1}{\lambda_1(\mathcal{L}_F)},
\end{equation}
which measures the average pinning control effectiveness across all pinning sizes up to the maximum fraction \( p_{\max} \). Smaller \( \omega \) values indicate higher average efficiency. The metric \( \delta \) is defined as
\begin{equation}
    \delta = \frac{1}{\lambda_1\left(\mathcal{L}_F\right)},
\end{equation}
corresponding to the performance at the largest tested pinning fraction \( p_{\max} \). Since \( 1/\lambda_1 \) is monotonically non-increasing with \( p \), this value coincides with the minimum over the range \( p \in [0, p_{\max}] \).

We further define relative improvement ratios of our methods over the best-performing suboptimal baseline as
\begin{equation}
    \Delta_\omega = \frac{\omega_{\text{SubOpt}} - \omega_{\text{OurBest}}}{\omega_{\text{SubOpt}}} \times 100\%,
\end{equation}
and
\begin{equation}
    \Delta_{\delta} = \frac{{\delta}_{\text{SubOpt}} - {\delta}_{\text{OurBest}}}{\delta_{\text{SubOpt}}} \times 100\%.
\end{equation}

As shown in Table~\ref{tab:results_ksat}, our Algorithm \( \mathcal{A}_2 \) outperforms all baselines in both metrics.
Across networks with six different \( k_{\text{sat}} \) values, the average improvement in \( \omega \) is \( \bar{\Delta}_{\omega_{\mathcal{A}_2}} = 70.55\% \), while that in \( \delta \) is \( \bar{\Delta}_{\delta_{\mathcal{A}_2}} = 65.13\% \). Moreover, both \( \omega \) and \( \delta \) decrease with increasing \( k_{\text{sat}} \), indicating that stronger low-degree saturation enhances pinning performance, in agreement with the trends observed in Figure~\ref{fig:results1}.

Beyond highlighting the advantages of our algorithms, the results in Figure~\ref{fig:results1} also expose a critical weakness of the heuristic baselines. As shown in Figure~\ref{fig:results1}, their \( 1/\lambda_1(\mathcal{L}_F) \) curves rapidly flatten, indicating that \(\lambda_1\) quickly reaches its theoretical bound—given by the minimum degree among free nodes—and cannot be further improved. For \( k_{\text{sat}} \le 100 \), this bound equals $1$ since the free set (selected by centrality strategies) contains degree-$1$ nodes, causing \(\lambda_1\) to converge to $1$. For \( k_{\text{sat}} = 120 \), all nodes have degree \(>1\), so \(\lambda_1\) converges to a value slightly below $1$. In both cases, once the bound is reached, further pinning yields no performance gain, making this plateau effect a pervasive limitation of traditional heuristics.

\subsubsection{High-degree cutoff}

We next examine the effect of high-degree cutoff on pinning control performance by varying the parameter \( k_{\text{cut}} \) in Equation~\eqref{eq:degree_distribution} while fixing \( a = 1 \), \( \gamma = 1.5 \), and \( k_{\text{sat}} = 20 \). Six cutoff levels are considered: \( k_{\text{cut}} = 100, 200, 300, 400, 500, 600 \).

As shown in Fig.~\ref{fig:results2}, both proposed algorithms consistently outperform all baselines across all cutoff levels. Increasing \( k_{\text{cut}} \)—which increases the likelihood of very high-degree nodes—leads to a mild increase in oscillation frequency and a gradual reduction in oscillation amplitude with growing pinning fraction \(p\). This trend indicates that a strong high-degree cutoff limits the ability to exploit hub nodes for enhancing synchronizability.

From Table~\ref{tab:results_kcut}, as \( k_{\text{cut}} \) increases, the values of both \( \omega \) and \( \delta \) decrease for all methods, reflecting an overall improvement in controllability. Our algorithms achieve the smallest values among all strategies, and their improvement ratios \(\Delta_{\omega}\) and \(\Delta_{\delta}\) also increase with \( k_{\text{cut}} \), showing that they can more effectively leverage the presence of high-degree nodes. For example, Algorithm~\(\mathcal{A}_2\) attains \(\delta=0.1020\) and \(\Delta_{\omega}=61.05\%\) at \(k_{\text{cut}}=600\), compared to \(\delta=0.1791\) and \(\Delta_{\omega}=45.56\%\) at \(k_{\text{cut}}=100\). Averaged over all six networks, the gains are \(\bar{\Delta}_{\omega_{\mathcal{A}_2}} = 57.10\%\) for Pinning Efficiency, and \(\bar{\Delta}_{\delta_{\mathcal{A}_2}} = 87.21\%\) for End-Point Pinning Effectiveness, confirming the substantial advantage of our approach.

These results are fully consistent with Theorem~\ref{thm:swap}, which predicts that incorporating even a small fraction of high-degree nodes into the pinning set can significantly reduce \(1/\lambda_1(\mathcal{L}_F)\) and thus enhance control effectiveness.

\begin{figure*}[htbp]
    \centering
    \begin{subfigure}{0.3\textwidth}
        \centering
        \includegraphics[width=\textwidth, height=0.73\textwidth]{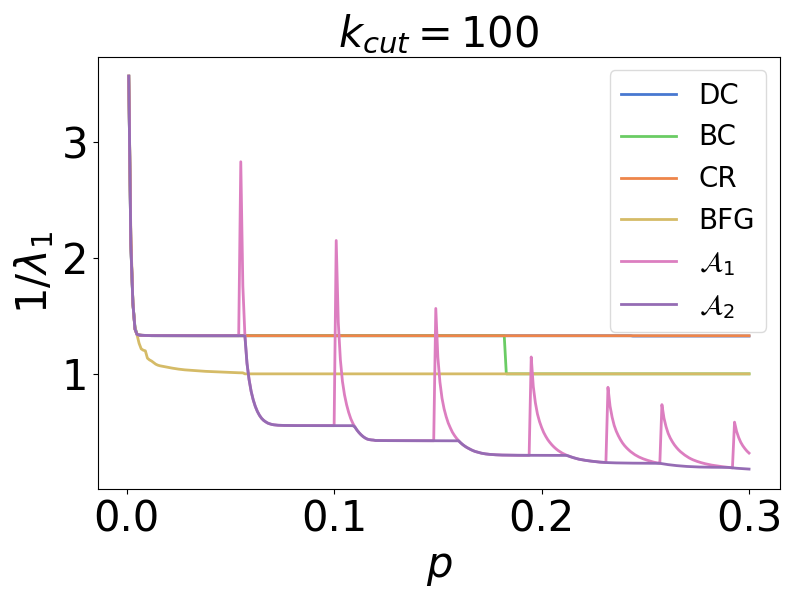}
        \caption{}
    \end{subfigure}
    \begin{subfigure}{0.3\textwidth}
        \centering
        \includegraphics[width=\textwidth, height=0.73\textwidth]{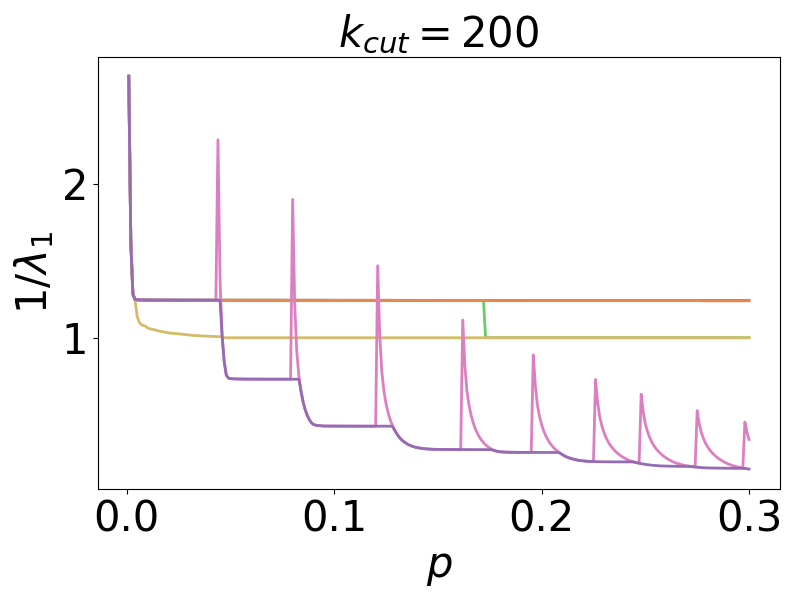}
        \caption{}
    \end{subfigure}
    \begin{subfigure}{0.3\textwidth}
        \centering
        \includegraphics[width=\textwidth, height=0.73\textwidth]{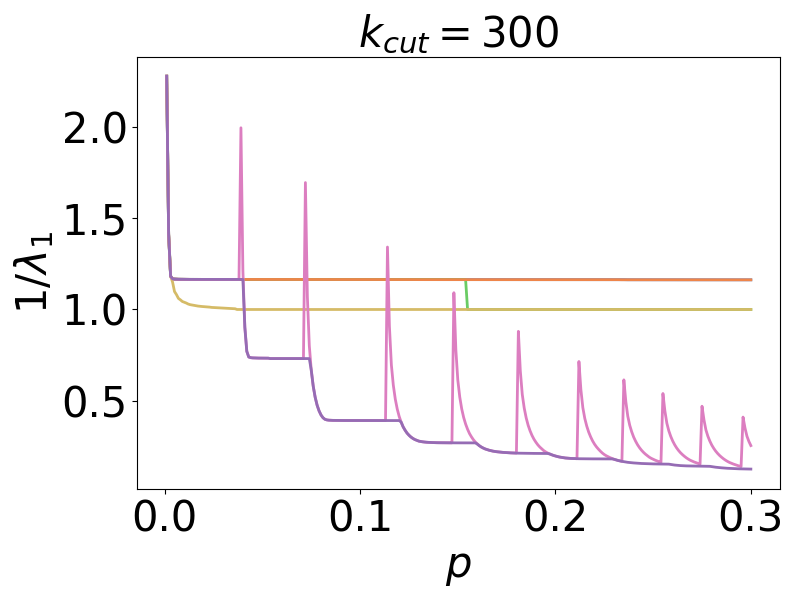}
        \caption{}
    \end{subfigure}
    \begin{subfigure}{0.3\textwidth}
        \centering
        \includegraphics[width=\textwidth, height=0.73\textwidth]{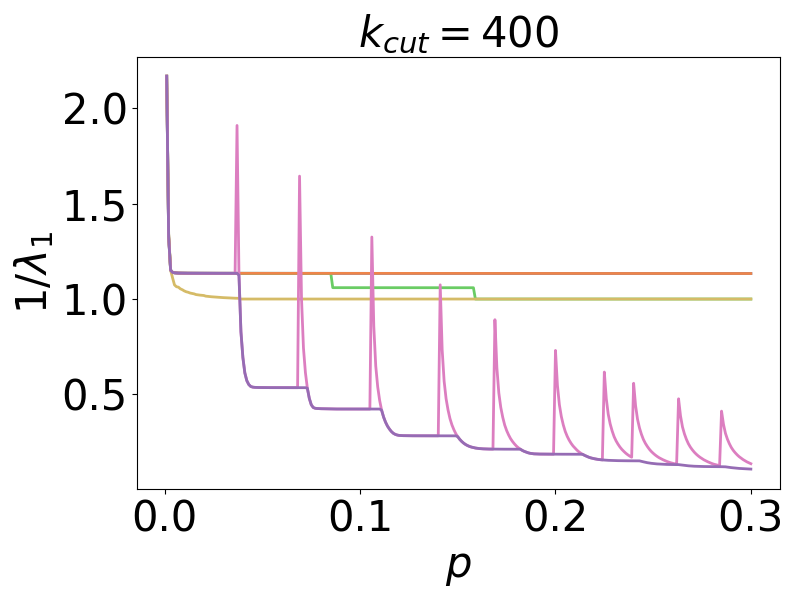}
        \caption{}
    \end{subfigure}
    \begin{subfigure}{0.3\textwidth}
        \centering
        \includegraphics[width=\textwidth, height=0.73\textwidth]{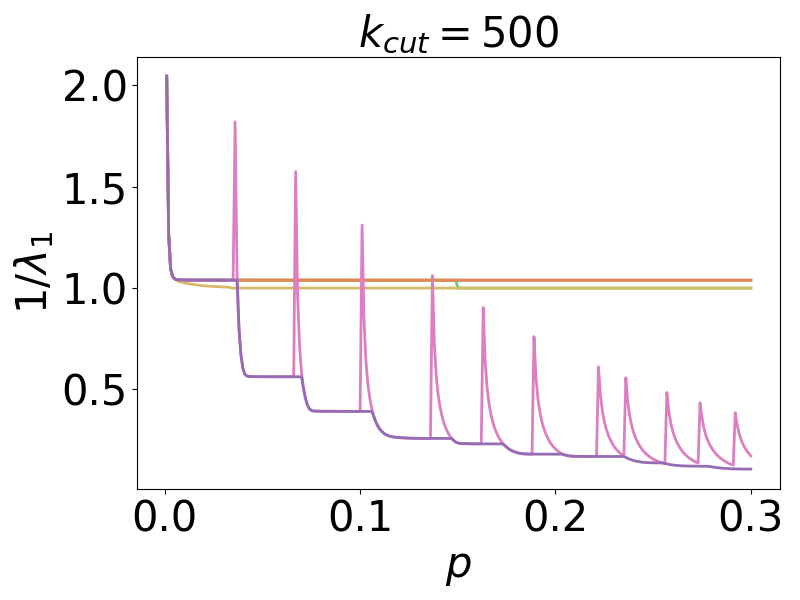}
        \caption{}
    \end{subfigure}
    \begin{subfigure}{0.3\textwidth}
        \centering
        \includegraphics[width=\textwidth, height=0.73\textwidth]{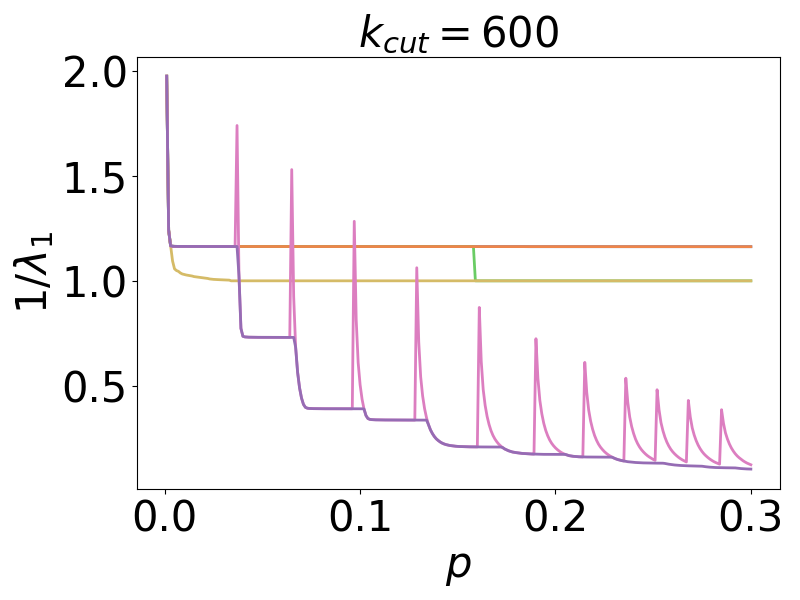}
        \caption{}
    \end{subfigure}
    \caption{Effect of high-degree cutoff on pinning performance for configuration networks with \(k_{\text{cut}}=100,200,300,400,500,600\).}
    \label{fig:results2}
\end{figure*}

\begin{table*}[htbp]
    \caption{Pinning Efficiency \( \omega \) and End-Point Pinning Effectiveness \( \delta \) under  varying \(k_{\text{cut}}\). Best results are in bold.}
    \centering
    \renewcommand{\arraystretch}{1.2}
    \resizebox{\textwidth}{!}{
    \large  
    \begin{tabular}{c|cc|cc|cc|cc|cc|cc|c|c}
        \toprule
        \multirow{2}{*}{$k_{cut}$} & \multicolumn{2}{c|}{DC} & \multicolumn{2}{c|}{BC} & \multicolumn{2}{c|}{CR} & \multicolumn{2}{c|}{BFG} & \multicolumn{2}{c|}{$\mathcal{A}_1$} & \multicolumn{2}{c|}{$\mathcal{A}_2$} & \multicolumn{1}{c|}{$\Delta_{\omega}$} & \multicolumn{1}{c}{$\Delta_{\delta}$} \\
        & $\omega$ & $\delta$ & $\omega$ & $\delta$ & $\omega$ & $\delta$ & $\omega$ & $\delta$ & $\omega$ & $\delta$ & $\omega$ & $\delta$ & & \\
        \midrule
        100 & 1.3395 & 1.3259 & 1.2108 & 1.0000 & 1.3403 & 1.3285 & 1.0257 & 1.0000 & 0.6327 & 0.3171 & \textbf{0.5584} & \textbf{0.1791} & \textbf{45.56\%} & \textbf{82.09\%} \\
        200 & 1.2491 & 1.2415 & 1.1455 & 1.0000 & 1.2488 & 1.2399 & 1.0142 & 1.0000 & 0.5423 & 0.3388 & \textbf{0.4807} & \textbf{0.1465} & \textbf{52.60\%} & \textbf{85.35\%} \\
        300 & 1.1687 & 1.1628 & 1.0891 & 1.0000 & 1.1687 & 1.1628 & 1.0092 & 1.0000 & 0.4869 & 0.2527 & \textbf{0.4237} & \textbf{0.1236} & \textbf{58.02\%} & \textbf{87.64\%} \\
        400 & 1.1383 & 1.1336 & 1.0568 & 1.0000 & 1.1384 & 1.1336 & 1.0084 & 1.0000 & 0.4470 & 0.1372 & \textbf{0.3869} & \textbf{0.1090} & \textbf{61.63\%} & \textbf{89.10\%} \\
        500 & 1.0433 & 1.0388 & 1.0237 & 1.0000 & 1.0433 & 1.0388 & 1.0064 & 1.0000 & 0.4272 & 0.1722 & \textbf{0.3648} & \textbf{0.1073} & \textbf{63.75\%} & \textbf{89.27\%} \\
        600 & 1.1662 & 1.1628 & 1.0890 & 1.0000 & 1.1661 & 1.1628 & 1.0067 & 1.0000 & 0.4482 & 0.1226 & \textbf{0.3921} & \textbf{0.1020} & \textbf{61.05\%} & \textbf{89.80\%} \\
        \bottomrule
    \end{tabular}}
    \label{tab:results_kcut}
\end{table*}

\begin{table}[htbp]
    \caption{Information of real networks}
    \centering
    \setlength{\tabcolsep}{13pt}
    \small 
    \begin{tabular}{lcccc}
        \toprule
        \textbf{Network} & \textbf{N} & \textbf{M} & \textbf{\( \langle k \rangle \)} \\
        \midrule
        Email & 1133 & 5451 & 9.62 \\
        Econ & 2529 & 86768 & 68.62\\
        Jazz & 198 & 2742 & 27.70 \\
        \bottomrule
    \end{tabular}
    \label{tab:network}
\end{table}

\begin{figure*}[htbp]
    \centering
    \begin{subfigure}{0.3\textwidth}
        \centering
        \includegraphics[width=\textwidth, height=0.73\textwidth]{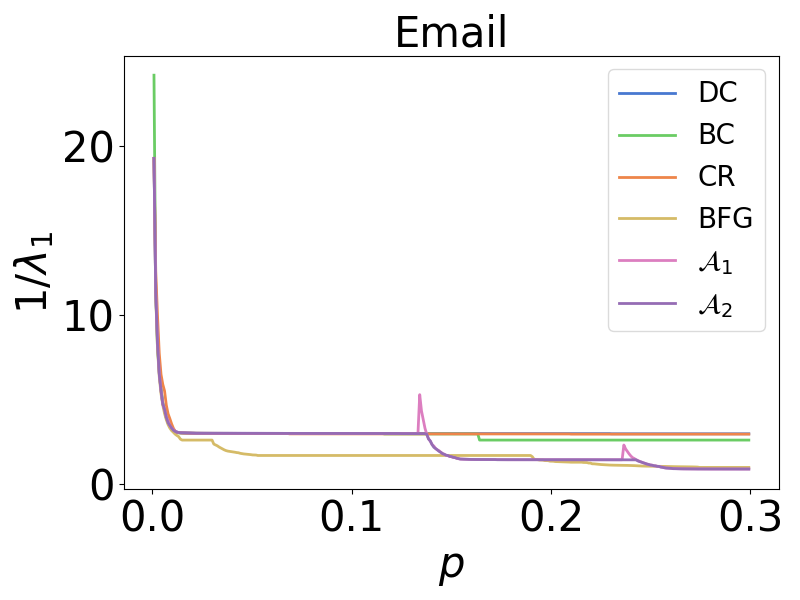}
        \caption{}
    \end{subfigure}
    \begin{subfigure}{0.3\textwidth}
        \centering
        \includegraphics[width=\textwidth, height=0.73\textwidth]{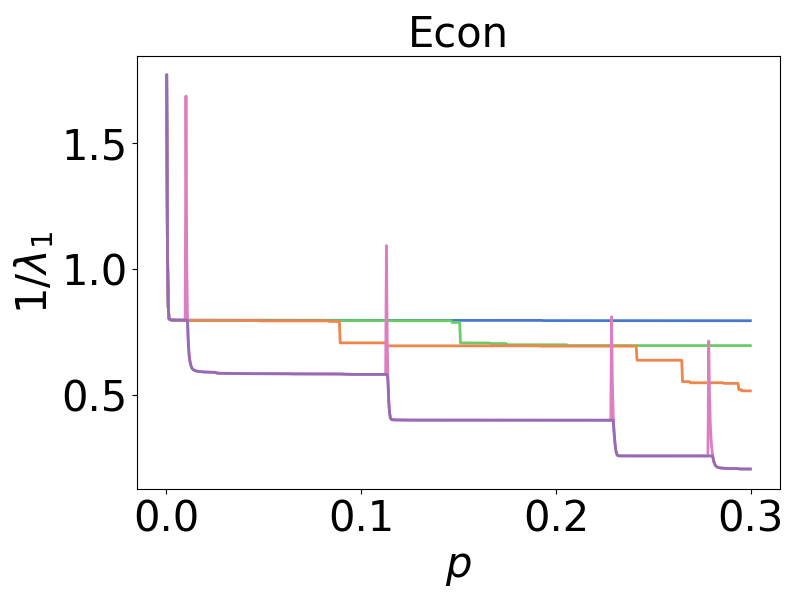}
        \caption{}
    \end{subfigure}
    \begin{subfigure}{0.3\textwidth}
        \centering
        \includegraphics[width=\textwidth, height=0.73\textwidth]{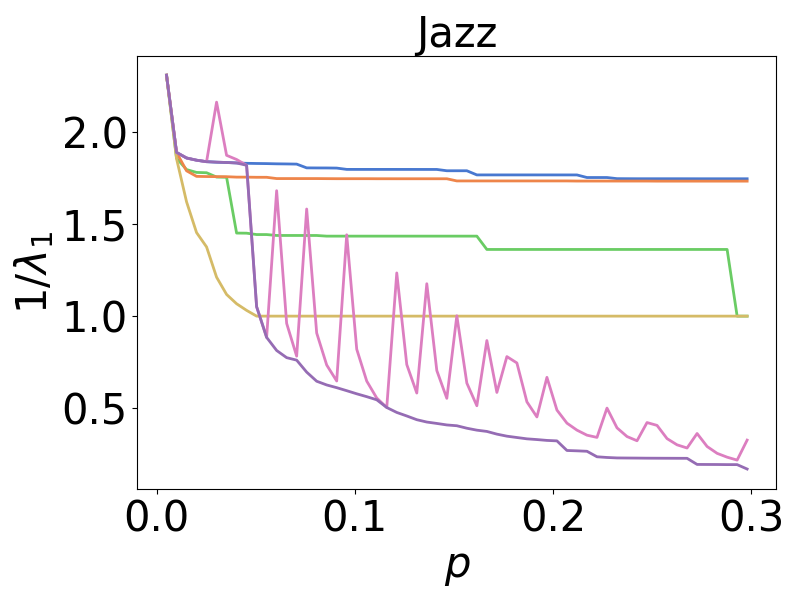}
        \caption{}
    \end{subfigure}
    \caption{Pinning control performance in three empirical networks. The proposed algorithms \(\mathcal{A}_1\) and \(\mathcal{A}_2\) consistently achieve lower \(1/\lambda_1(L_{N(1-p)})\) curves than all baselines across the entire pinning fraction \(p\), indicating superior synchronizability after pinning.}
    \label{fig:results4}
\end{figure*}

\begin{table*}[htbp]
    \caption{Comparison of \emph{Pinning Efficiency} \(\omega\) and \emph{End-Point Pinning Effectiveness} \(\delta\) in empirical networks. Best results are in bold.}
    \centering
    \renewcommand{\arraystretch}{1.2}
    \resizebox{\textwidth}{!}{
    \large 
    \begin{tabular}{c|cc|cc|cc|cc|cc|cc|c|c}
        \toprule
        \multirow{2}{*}{Network} & \multicolumn{2}{c|}{DC} & \multicolumn{2}{c|}{BC} & \multicolumn{2}{c|}{CR} & \multicolumn{2}{c|}{BFG} & \multicolumn{2}{c|}{$\mathcal{A}_1$} & \multicolumn{2}{c|}{$\mathcal{A}_2$} & \multicolumn{1}{c|}{$\Delta_{\omega}$} & \multicolumn{1}{c}{$\Delta_{\delta}$} \\
        & $\omega$ & $\delta$ & $\omega$ & $\delta$ & $\omega$ & $\delta$ & $\omega$ & $\delta$ & $\omega$ & $\delta$ & $\omega$ & $\delta$ & & \\
        \midrule
        Email & 3.1287 & 2.9898 & 2.9644 & 2.6180 & 3.1376 & 2.9697 & \textbf{1.7445} & 1.0000 & 2.2539 & 0.9053 & 2.2296 & \textbf{0.9052} & - & \textbf{9.48\%} \\
        Econ  & 0.7944 & 0.7944 & 0.7482 & 0.6955 & 0.7044 & 0.5155 & - & - & 0.4477 & 0.2048 & \textbf{0.4433} & \textbf{0.2048} & \textbf{37.07\%} & \textbf{60.27\%} \\
        Jazz  & 1.7946 & 1.1628 & 1.4407 & 1.0000 & 1.7536 & 1.7328 & 1.0685 & 1.0000 & 0.8374 & 0.3270 & \textbf{0.6355} & \textbf{0.1699} & \textbf{40.52\%} & \textbf{83.01\%} \\
        \bottomrule
    \end{tabular}}
    \label{tab:results_real}
\end{table*}

\subsection{Real-World Networks}

We further validated our methods on three real-world networks exhibiting power-law-like degree distribution characteristics~\cite{fan2021cycles}, summarized in Table~\ref{tab:network}. The datasets span different domains and structural scales, including the Email communication network~\cite{guimera2003self}, the Econ trade network~\cite{nr}, and the Jazz musician collaboration network~\cite{gleiser2003community}.

As shown in Fig.~\ref{fig:results4} (see also Table~\ref{tab:results_real}), both proposed algorithms \(\mathcal{A}_1\) and \(\mathcal{A}_2\) consistently outperform all baselines across the three empirical networks. The curves of \(1/\lambda_1(\mathcal{L}_F)\) are uniformly lower for our methods over the entire pinning range \(p\), and the aggregate metrics—\(\omega\) and \(\delta\)—are also consistently smaller, indicating stronger post-pinning synchronizability than competing strategies. Notably, for the Econ network, we did not include the BFG baseline due to its excessive computational cost.

These empirical results demonstrate the robustness and practical extensibility of our approach: by exploiting degree-distribution information under the annealed-approximation framework, the proposed algorithms provide reliable guidance for selecting the optimal driver-node set for pinning control in real-world networks.

\section{Discussion}
The study provides new mechanistic insights into optimal pinning control through an analytical framework based on the annealed approximation. The derived degree-based selection rule reveals that the globally optimal pinning configuration emerges from a degree-layered structure, balancing the benefits of pinning low-degree nodes—which can produce a larger relative increase in $\lambda_1(L_F)$ in budget-limited regimes—and the connectivity advantages of high-degree nodes. This finding fundamentally challenges the prevailing hub-dominated heuristics, showing that optimal strategies can involve nontrivial combinations of low- and high-degree nodes.

A key theoretical discovery is the chaotic dependence of the optimal driver set on its cardinality: marginal changes in the pinning budget can trigger abrupt reconfigurations of the selected nodes and significant fluctuations in control effectiveness. Unlike traditional heuristics that miss this phenomenon, our framework explains its structural origin via the degree-layered composition and offers constructive algorithms (\(\mathcal{A}_1,\mathcal{A}_2\)). This insight turns chaotic dependence from a source of fragility into a guide for robust and adaptive strategy design.

Beyond theory, the proposed algorithms are computationally efficient and empirically validated on both synthetic and real networks, demonstrating scalability and reproducibility across diverse network settings. They consistently reduce Pinning Efficiency \(\omega\) and End-Point Pinning Effectiveness \(\delta\) relative to baselines, demonstrating that exploiting degree-distribution features leads to substantial and reproducible gains. By bridging topological intuition with spectral control metrics, this work establishes a general design principle applicable to a broad spectrum of networked systems.

\section{Conclusion and Future Work} \label{sec:conclusion}
This work addresses the fundamental problem of globally optimal driver-node selection for pinning control by leveraging a degree-based mean-field (annealed) approximation, which transforms a combinatorial optimization challenge into a tractable analytical mechanism. The study offers three principal outcomes. First, it provides a theoretical mechanism in the form of a degree-based selection rule that reduces the problem to a structured degree-layered decision, yielding closed-form optimality conditions under the annealed model, and clarifying the role of both low- and high-degree nodes in enhancing synchronizability. Second, it formally identifies and analyzes the chaotic dependence of the global optimum on set cardinality, and presents constructive algorithms that yield the exact global optimum within the annealed approximation framework. Our analysis not only explains why low-degree nodes can be beneficial in certain regimes but also reveals the chaotic dependence of the optimal set on its cardinality, a phenomenon that introduces both opportunities and challenges for control design. Third, it delivers practical guidance by quantifying how $k_{\text{sat}}$, $k_{\text{cut}}$, and $\gamma$ systematically influence spectral controllability, offering concrete rules for selecting driver nodes in real systems.

Extensive experiments on configuration-model ensembles and empirical networks validate these theoretical predictions and demonstrate consistent, often substantial, improvements over heuristic baselines. These results establish a principled link between degree-distribution structure and controllability, and provide a practical toolkit for driver-node placement.

The insights obtained here have cross-domain relevance. In power systems~\cite{dorfler2013synchronization,motter2013spontaneous}, they may guide optimal actuator placement for frequency synchronization; in brain networks~\cite{bassett2017network,muldoon2016stimulation}, they may inform stimulation targeting for neurological disorders; in biochemical systems~\cite{alon2009design,cornelius2013realistic}, they may identify key regulatory points for metabolic control. The analytical structure developed in this work thus serves as a generalizable design principle for complex network control.

Several directions emerge for future work. Extending the theoretical framework beyond the annealed approximation to quenched networks with clustering, community structure, or degree correlations would allow for application to more realistic topologies. Adapting the degree-based rule to nonlinear or time-delayed dynamics, such as adaptive Kuramoto models or nonlinear consensus protocols, remains an open challenge. Building on the current global solution, robust and adaptive strategies should be developed to maintain high performance under perturbations to network topology or budget constraints, including dynamic and multilayer networks. It is also important to generalize the framework to directed, weighted, and multilayer graphs where degree-like measures and interlayer couplings interact. From a computational perspective, scalable and distributed implementations could enable deployment in billion-node networks, while incorporating cost or capacity constraints into the optimization. 

In summary, this work provides a rigorous theoretical foundation, an analytically tractable solution, and empirically validated algorithms for optimal pinning control. By linking degree-distribution features to controllability, it advances both the mechanistic understanding and the practical design of effective control strategies in complex networks.

\section*{Acknowledgment}
This work was supported by the National Natural Science Foundation of China (Grant No. T2293771, 62503447), the STI 2030 Major Projects (Grant No. 2022ZD0211400), the China Postdoctoral Science Foundation (Grant No. 2024M763131), the Postdoctoral Fellowship Program of CPSF (Grant No. GZC20241653), and the New Cornerstone Science Foundation through the XPLORER PRIZE.




\bibliographystyle{IEEEtran}
\bibliography{main}

@PREAMBLE{
 "\providecommand{\noopsort}[1]{}" 
 # "\providecommand{\singleletter}[1]{#1}%" 
}

@article{grigoriev1997pinning,
  author = {R. O. Grigoriev and M. C. Cross and H. G. Schuster},
  title = {Pinning Control of Spatiotemporal Chaos},
  journal = {Phys. Rev. Lett.},
  volume = {79},
  number = {15},
  pages = {2795--2798},
  year = {1997},
  doi = {10.1103/PhysRevLett.79.2795}
}

@article{wang2002pinning,
  author = {X. F. Wang and G. Chen},
  title = {Pinning control of scale-free dynamical networks},
  journal = {Physica A: Statistical Mechanics and its Applications},
  volume = {310},
  number = {3--4},
  pages = {521--531},
  year = {2002},
  doi = {10.1016/S0378-4371(02)00772-0}
}

@article{li2004pinning,
  author = {X. Li and X. Wang and G. Chen},
  title = {Pinning a complex dynamical network to its equilibrium},
  journal = {IEEE Transactions on Circuits and Systems I: Regular Papers},
  volume = {51},
  number = {10},
  pages = {2074--2087},
  year = {2004}
}

@article{liu2007pinning,
  author = {Z. Liu and Z. Chen and Z. Yuan},
  title = {Pinning control of weighted general complex dynamical networks with time delay},
  journal = {Physica A: Statistical Mechanics and its Applications},
  volume = {375},
  number = {1},
  pages = {345--354},
  year = {2007}
}

@article{zhou2008adaptive,
  author = {J. Zhou and J. Lu and J. Lü},
  title = {Pinning adaptive synchronization of a general complex dynamical network},
  journal = {Automatica},
  volume = {44},
  number = {4},
  pages = {996--1003},
  year = {2008}
}

@article{pecora1998msf,
  author = {L. M. Pecora and T. L. Carroll},
  title = {Master stability functions for synchronized coupled systems},
  journal = {Physical Review Letters},
  volume = {80},
  number = {10},
  pages = {2109},
  year = {1998}
}

@article{moradiamani2017influential,
  author = {A. Moradi Amani and M. Jalili and X. Yu and L. Stone},
  title = {Finding the Most Influential Nodes in Pinning Controllability of Complex Networks},
  journal = {IEEE Trans. Circuits Syst. II},
  volume = {64},
  number = {6},
  pages = {685--689},
  year = {2017},
  doi = {10.1109/TCSII.2016.2601565}
}

@inproceedings{yu2013optimalvertex,
  author = {W. Yu and J. Lu and X. Yu and G. Chen},
  title = {A step forward to pinning control of complex networks: Finding an optimal vertex to control},
  booktitle = {2013 9th Asian Control Conference (ASCC)},
  organization = {IEEE},
  address = {Istanbul, Turkey},
  pages = {1--6},
  year = {2013},
  doi = {10.1109/ASCC.2013.6606272}
}

@article{liu2024uniformity,
  author = {J. Liu and Z. Wu and Q. Xin and M. Yu and L. Liu},
  title = {Topology uniformity pinning control for multi-agent flocking},
  journal = {Complex Intell. Syst.},
  volume = {10},
  number = {2},
  pages = {2013--2027},
  year = {2024},
  doi = {10.1007/s40747-023-01253-7}
}

@article{liu2021spectralpinning,
  author = {H. Liu and X. Xu and J.-A. Lu and G. Chen and Z. Zeng},
  title = {Optimizing Pinning Control of Complex Dynamical Networks Based on Spectral Properties of Grounded Laplacian Matrices},
  journal = {IEEE Trans. Syst. Man Cybern, Syst.},
  volume = {51},
  number = {2},
  pages = {786--796},
  year = {2021},
  doi = {10.1109/TSMC.2018.2882620}
}

@article{pan2021optimal,
  author = {L. Pan and W. Wang and L. Tian and Y.-C. Lai},
  title = {Optimal networks for dynamical spreading},
  journal = {Physical Review E},
  volume = {103},
  number = {1},
  pages = {012302},
  year = {2021}
}

@article{barabasi2013networkscience,
  author = {A.-L. Barabási},
  title = {Network science},
  journal = {Philosophical Transactions of the Royal Society A: Mathematical, Physical and Engineering Sciences},
  volume = {371},
  number = {1987},
  pages = {20120375},
  year = {2013}
}

@article{fan2021cycles,
  author = {T. Fan and L. Lü and D. Shi and T. Zhou},
  title = {Characterizing cycle structure in complex networks},
  journal = {Commun Phys},
  volume = {4},
  number = {1},
  pages = {272},
  year = {2021},
  doi = {10.1038/s42005-021-00781-3}
}

@article{pastor2015epidemic,
  title={Epidemic processes in complex networks},
  author={Pastor-Satorras, Romualdo and Castellano, Claudio and Van Mieghem, Piet and Vespignani, Alessandro},
  journal={Reviews of modern physics},
  volume={87},
  number={3},
  pages={925--979},
  year={2015},
  publisher={APS}
}

@article{ferreira2012epidemic,
  title={Epidemic thresholds of the susceptible-infected-susceptible model on networks: A comparison of numerical and theoretical results},
  author={Ferreira, Silvio C and Castellano, Claudio and Pastor-Satorras, Romualdo},
  journal={Physical Review E—Statistical, Nonlinear, and Soft Matter Physics},
  volume={86},
  number={4},
  pages={041125},
  year={2012},
  publisher={APS}
}

@article{albert1999diameter,
  title={Diameter of the world-wide web},
  author={Albert, R{\'e}ka and Jeong, Hawoong and Barab{\'a}si, Albert-L{\'a}szl{\'o}},
  journal={nature},
  volume={401},
  number={6749},
  pages={130--131},
  year={1999},
  publisher={Nature Publishing Group UK London}
}

@article{dorfler2013synchronization,
  title={Synchronization in complex oscillator networks and smart grids},
  author={D{\"o}rfler, Florian and Chertkov, Michael and Bullo, Francesco},
  journal={Proceedings of the National Academy of Sciences},
  volume={110},
  number={6},
  pages={2005--2010},
  year={2013},
  publisher={National Acad Sciences}
}

@article{arenas2008synchronization,
  title={Synchronization in complex networks},
  author={Arenas, Alex and Díaz-Guilera, Albert and Kurths, Jürgen and Moreno, Yamir and Zhou, Changsong},
  journal={Physics Reports},
  volume={469},
  number={3},
  pages={93--153},
  year={2008},
  publisher={Elsevier},
  doi={10.1016/j.physrep.2008.09.002}
}

@article{wang2008adaptive,
  title={Adaptive synchronization of weighted complex dynamical networks through pinning},
  author={Wang, L and Dai, HP and Dong, H and Cao, YY and Sun, YX},
  journal={The European Physical Journal B},
  volume={61},
  pages={335--342},
  year={2008},
  publisher={Springer}
}

@article{song2009pinning,
  title={On pinning synchronization of directed and undirected complex dynamical networks},
  author={Song, Qiang and Cao, Jinde},
  journal={IEEE Transactions on Circuits and Systems I: Regular Papers},
  volume={57},
  number={3},
  pages={672--680},
  year={2009},
  publisher={IEEE}
}

@article{yu2009distributed,
  title={Distributed consensus filtering in sensor networks},
  author={Yu, Wenwu and Chen, Guanrong and Wang, Zidong and Yang, Wen},
  journal={IEEE Transactions on Systems, Man, and Cybernetics, Part B (Cybernetics)},
  volume={39},
  number={6},
  pages={1568--1577},
  year={2009},
  publisher={IEEE}
}

@article{mao2025perturbation,
  title={Perturbation-Based Pinning Control Strategy for Enhanced Synchronization in Complex Networks},
  author={Mao, Ziang and Fan, Tianlong and L{\"u}, Linyuan},
  journal={arXiv preprint arXiv:2504.00493},
  year={2025}
}

@article{gershgorin1931uber,
  title        = {Über die Abgrenzung der Eigenwerte einer Matrix},
  author       = {Gershgorin, Semen Aronovich},
  journal      = {Izvestiya Akademii Nauk SSSR, Seriya Matematicheskaya},
  volume       = {6},
  number       = {6},
  pages        = {749--754},
  year         = {1931},
  publisher    = {Steklov Mathematical Institute, Russian Academy of Sciences}
}

@article{guimera2003self,
  title={Self-similar community structure in a network of human interactions},
  author={Guimera, Roger and Danon, Leon and Diaz-Guilera, Albert and Giralt, Francesc and Arenas, Alex},
  journal={Physical review E},
  volume={68},
  number={6},
  pages={065103},
  year={2003},
  publisher={APS}
}

@article{gleiser2003community,
  title={Community structure in jazz},
  author={Gleiser, Pablo M and Danon, Leon},
  journal={Advances in complex systems},
  volume={6},
  number={04},
  pages={565--573},
  year={2003},
  publisher={World Scientific}
}

@inproceedings{nr,
     title={The Network Data Repository with Interactive Graph Analytics and Visualization},
     author={Ryan A. Rossi and Nesreen K. Ahmed},
     booktitle={AAAI},
     url={https://networkrepository.com},
     year={2015}
}

@article{belykh2004connection,
  title={Connection graph stability method for synchronized coupled chaotic systems},
  author={Belykh, Vladimir N and Belykh, Igor V and Hasler, Martin},
  journal={Physica D: nonlinear phenomena},
  volume={195},
  number={1-2},
  pages={159--187},
  year={2004},
  publisher={Elsevier}
}

@article{chen2011bidirectionally,
  title={Bidirectionally coupled synchronization of the generalized Lorenz systems},
  author={Chen, Juan and Lu, Jun-an and Wu, Xiaoqun},
  journal={Journal of Systems Science and Complexity},
  volume={24},
  number={3},
  pages={433--448},
  year={2011},
  publisher={Springer}
}

@article{liu2013coupling,
  title={Coupling strength allocation for synchronization in complex networks using spectral graph theory},
  author={Liu, Hui and Cao, Ming and Wu, Chai Wah},
  journal={IEEE Transactions on Circuits and Systems I: Regular Papers},
  volume={61},
  number={5},
  pages={1520--1530},
  year={2013},
  publisher={IEEE}
}

@article{liu2015synchronization,
  title={Synchronization in directed complex networks using graph comparison tools},
  author={Liu, Hui and Cao, Ming and Wu, Chai Wah and Lu, Jun-An and Tse, Chi K},
  journal={IEEE Transactions on Circuits and Systems I: Regular Papers},
  volume={62},
  number={4},
  pages={1185--1194},
  year={2015},
  publisher={IEEE}
}

@book{wu2007synchronization,
  title={Synchronization in complex networks of nonlinear dynamical systems},
  author={Wu, Chai Wah},
  year={2007},
  publisher={World scientific}
}

@article{motter2013spontaneous,
  title={Spontaneous synchrony in power-grid networks},
  author={Motter, Adilson E and Myers, Seth A and Anghel, Marian and Nishikawa, Takashi},
  journal={Nature Physics},
  volume={9},
  number={3},
  pages={191--197},
  year={2013},
  publisher={Nature Publishing Group UK London}
}

@article{bassett2017network,
  title={Network neuroscience},
  author={Bassett, Danielle S and Sporns, Olaf},
  journal={Nature neuroscience},
  volume={20},
  number={3},
  pages={353--364},
  year={2017},
  publisher={Nature Publishing Group US New York}
}

@article{muldoon2016stimulation,
  title={Stimulation-based control of dynamic brain networks},
  author={Muldoon, Sarah Feldt and Pasqualetti, Fabio and Gu, Shi and Cieslak, Matthew and Grafton, Scott T and Vettel, Jean M and Bassett, Danielle S},
  journal={PLoS computational biology},
  volume={12},
  number={9},
  pages={e1005076},
  year={2016},
  publisher={Public Library of Science San Francisco, CA USA}
}

@article{alon2009design,
  title={Design principles of biological circuits},
  author={Alon, Uri},
  journal={Biophysical Journal},
  volume={96},
  number={3},
  pages={15a--15a},
  year={2009}
}

@article{cornelius2013realistic,
  title={Realistic control of network dynamics},
  author={Cornelius, Sean P and Kath, William L and Motter, Adilson E},
  journal={Nature communications},
  volume={4},
  number={1},
  pages={1942},
  year={2013},
  publisher={Nature Publishing Group UK London}
}

%


\begin{IEEEbiography}[{\includegraphics[width=1in,height=1.25in,clip,keepaspectratio]{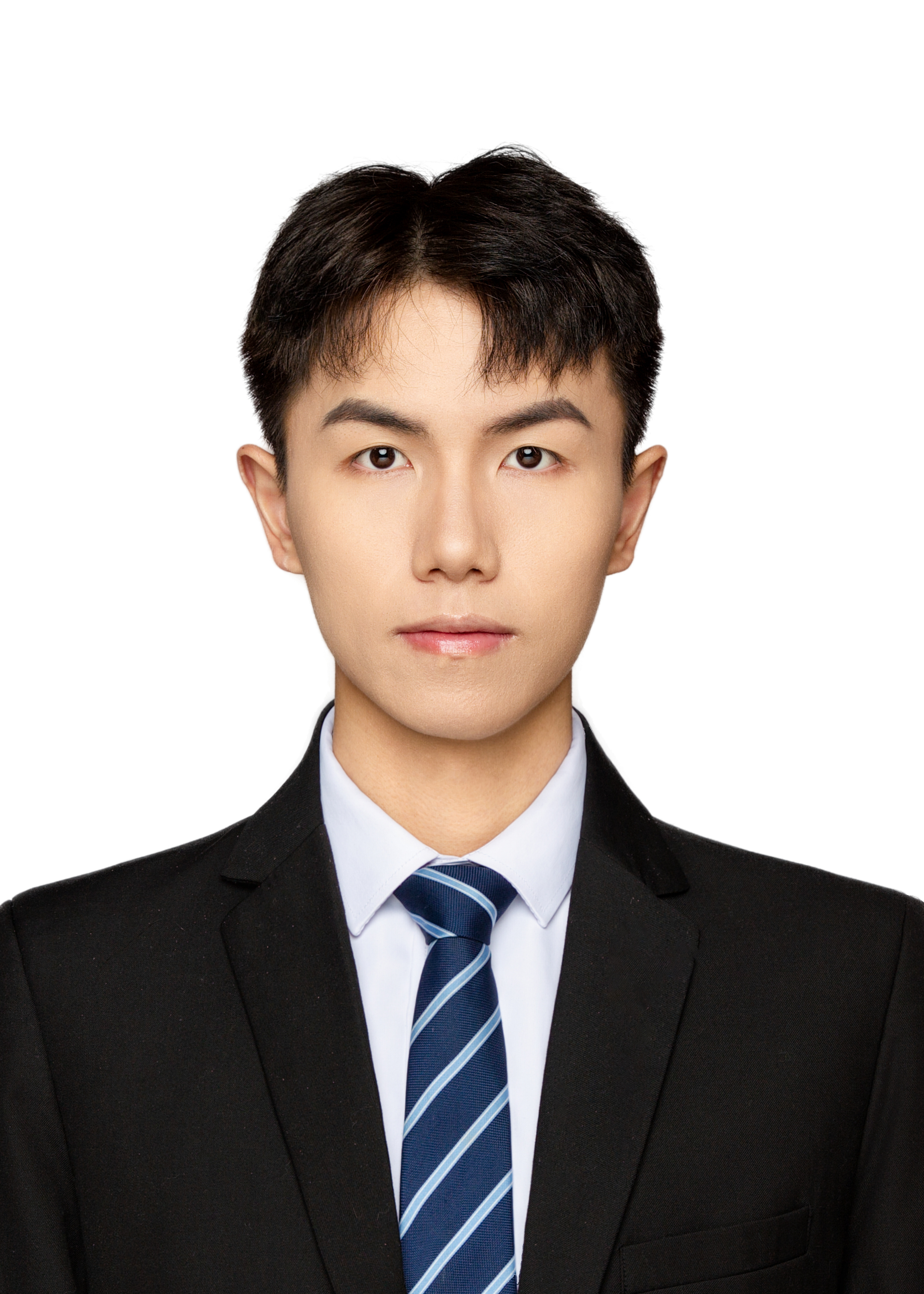}}]{Qingyang Liu}
received the B.Eng. degree in information security from University of Science and Technology of China, Hefei, China, in 2023. He is currently pursuing the Ph.D degree with the School of Cyber Science and Technology, University of Science and Technology of China, Hefei, China. 

His current research interests include synchronization and pinning control of complex networks.
\end{IEEEbiography}

\begin{IEEEbiography}[{\includegraphics[width=1in,height=1.25in,clip,keepaspectratio]{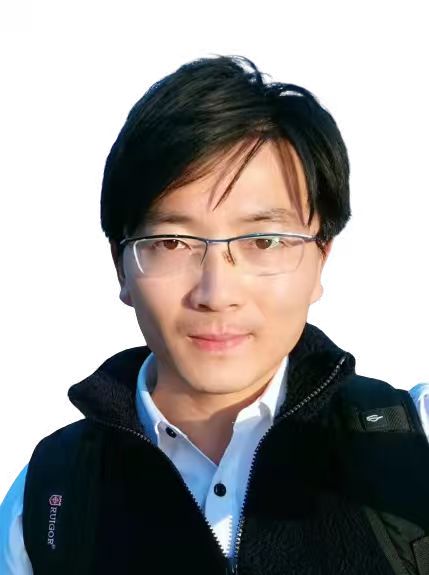}}]
{Tianlong Fan} received the Ph.D. degree in theoretical physics from the University of Fribourg, Fribourg, Switzerland, in 2023. He is currently a postdoctoral researcher in the School of Cyber Science and Technology, University of Science and Technology of China, Hefei, China. 

His current research interests include complex networks and theoretical physics.
\end{IEEEbiography}

\begin{IEEEbiography}[{\includegraphics[width=1in,height=1.25in,clip,keepaspectratio]{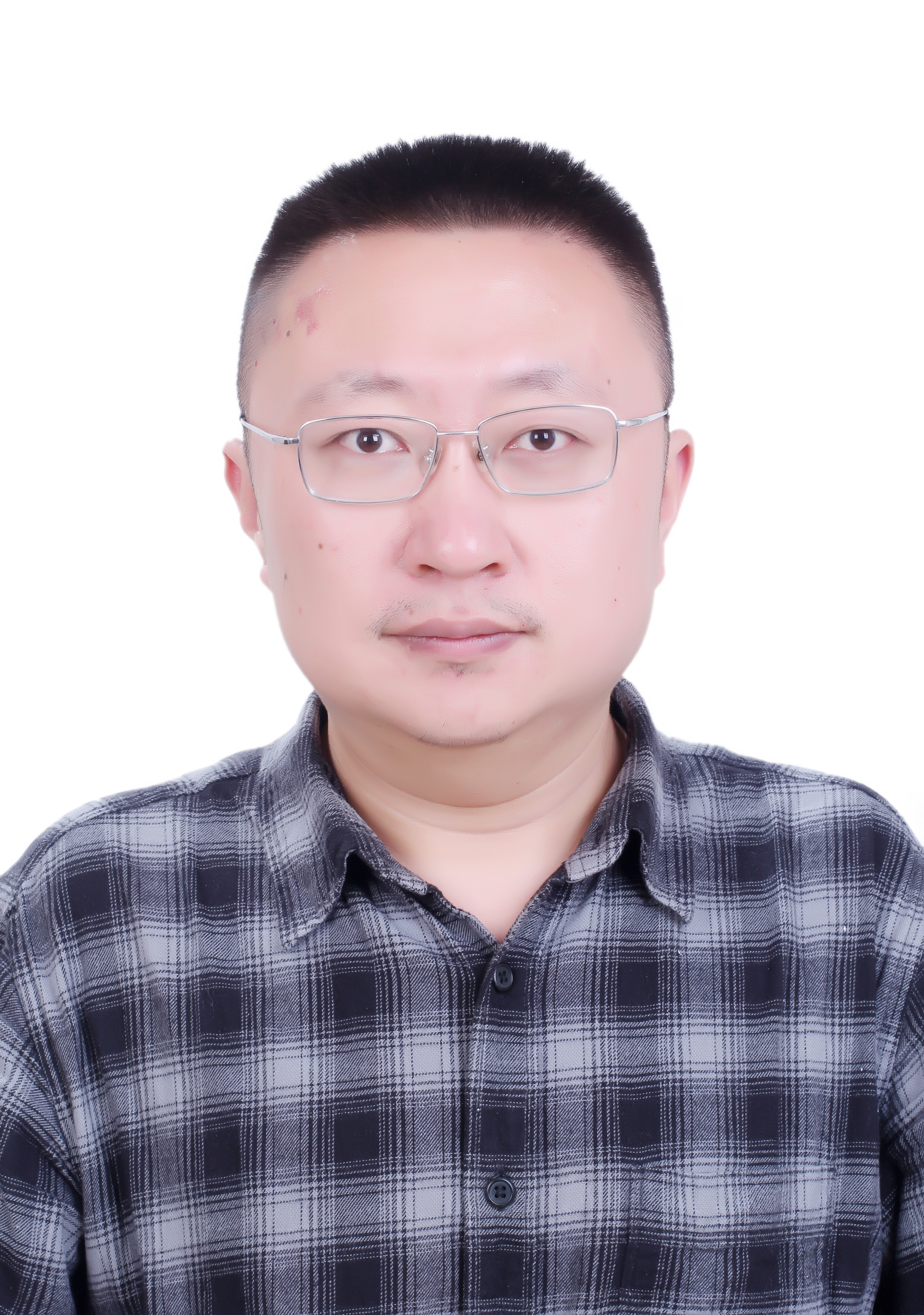}}]
{Liming Pan} received the Ph.D. degree in computer science from University of Electronic Science and Technology of China, Chengdu, China, in 2019. He is currently a associate researcher in School of Cyber Science, University of Science and Technology of China.

His current research interests include complex systems and artificial intelligence.
\end{IEEEbiography}

\begin{IEEEbiography}[{\includegraphics[width=1in,height=1.25in,clip,keepaspectratio]{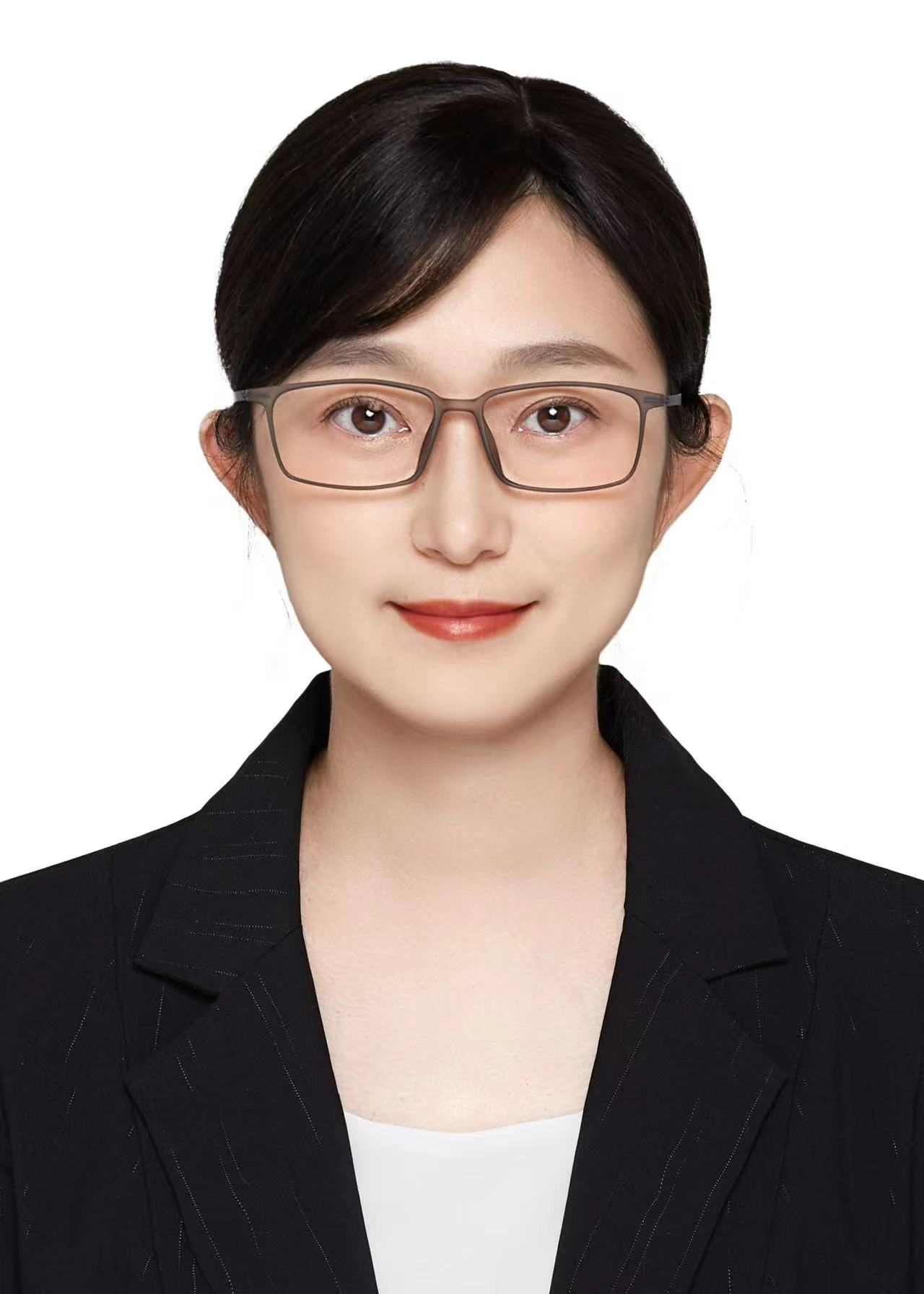}}]
{Linyuan {Lü}}
(Senior Member, IEEE) received the Ph.D. degree in theoretical physics from the Universite de Fribourg, Fribourg, Switzerland, in 2012. She is currently a Professor with the University of Science and Technology of China, Hefei, China. She has published over 80 articles in leading journals, including National Science Review, Physics Reports, and Nature Communications. 

Her current research interests include complex systems and higher-order network analysis.
\end{IEEEbiography}




\end{document}